\documentclass[cmp]{svjour}
\usepackage{graphicx,color}
\usepackage{amssymb,amsmath}
\usepackage{bm}
\usepackage[colorlinks=true]{hyperref}
\usepackage{enumerate}

\newcommand{\Ignore}[1]{}

\newcommand{\ket}[1]{\left\vert #1\right\rangle}
\newcommand{\bra}[1]{\left\langle #1\right\vert}
\newcommand{\braket}[2]{\langle#1\vert #2\rangle}
\newcommand{\ketbra}[2]{\vert#1\rangle\langle#2\vert}

\def\e{\mathrm{e}}

\def\d{\mathrm{d}}
\def\D{\mathrm{\Delta}}
\def\tr{\mathrm{tr}}
\def\deg{\mathcal{D}}

\def\H{\mathcal{H}}
\def\C{\mathbb{C}}
\def\I{\bm{I}}
\def\Z{\mathbb{Z}}
\def\R{\mathbb{R}}

\begin{document}

\title{Feynman graphs and the large dimensional limit of multipartite entanglement}
\author{}
\institute{}
\author{Sara Di Martino \inst{1}\and Paolo Facchi \inst{2,3}\and Giuseppe Florio \inst{4,3}}
\institute{F{\'i}sica Te{\`o}rica: Informaci{\'o} i Fen{\`o}mens Qu{\`a}ntics, Departament de F{\'i}sica, Universitat Aut{\`o}noma de Barcelona, 08193 Bellaterra (Barcelona), Spain \email{sara.dimartino@uniba.it}\and
Dipartimento di Fisica and MECENAS, Universit\`a di Bari, I-70126  Bari, Italy
\and
INFN, Sezione di Bari, I-70126 Bari, Italy
\and
Dipartimento di Meccanica, Matematica e Management, Politecnico di Bari,
  Via E.~Orabona 4,I--70125 Bari, Italy}
%\date{\today}
%\communicated{ }
\maketitle

\begin{abstract}
We are interested in the properties of multipartite entanglement of a system composed by $n$ $d$-level parties (qudits).

Focussing our attention on pure states we want to tackle the problem of the maximization of the entanglement for such systems. In particular we effort the problem trying to minimize the purity of the system. It has been shown that not for all systems this function can reach its lower bound, however it can be proved that for all values of $n$ a $d$ can always be found such that the lower bound can be reached. 

In this paper we examine the high-temperature expansion of the distribution function of the bipartite purity over all balanced bipartition considering its optimization problem as a problem of statistical mechanics. In particular we prove that the series characterizing the expansion converges and we analyze the behavior of each term of the series as $d\to \infty$.
\end{abstract}


\tableofcontents

\section{Introduction}
Since its early origins~\cite{schr} entanglement has been considered as one of the most basic and intriguing features of quantum mechanics~\cite{peres}. During the years it has turned out to be a fundamental resource in quantum information~\cite{RevAmico,RevHoro,nielsenchuang} and has originated a large number of research topics in mathematical~\cite{additivity,relative,mezzadri1,mezzadri2,shorcmp,hayden,giraud1,giraud2} and applied science with the quantum teleportation technology~\cite{teleport} and quantum key distribution protocols~\cite{crypto1,crypto2,crypto3}. 

The characterization and quantification of quantum correlations is not a simple task.
Bipartite entanglement, i.e. the entanglement of two subsystems, denoted as $A$ and $\bar{A}$, is well understood and can be completely characterized, for instance, using the von Neumann entropy~\cite{geometry} or the entanglement of formation~\cite{conc2}. Another possible measure is the so-called {\it purity} (of the relevant subsystem) $\pi_A$. Given an initial pure state, one can obtain the reduced density matrix of subsystem $A$ performing a partial trace over the degrees of freedom of subsystem $\bar{A}$; the function $\pi_A$ measures how ``pure" is the reduced state, i.e. its ``distance" from a completely mixed state. The more entangled is the initial pure state, the smaller is the value of this purity. By definition, a maximally entangled (pure) state will be left in a completely mixed state after the partial trace.

On the other hand, multipartite entanglement is less understood and more elusive even if widely investigated~\cite{monogamy,wc,mw,bruss}. These difficulties are deeply rooted both in the exponentially (with the system size) large number of measurements needed for its complete characterization and in new phenomena emerging from the complex interactions among the parties. Obviuosly, the choice of a particular measure, the dimension of the Hilbert space of local parties, and the simmetries imposed on quantum states will have an influence on the result. A natural question is whether it is possible to find maximally entangled states in the multipartite scenario. For instance, in~\cite{gisin} Gisin \textit{et al.} characterize pure and symmetric maximally entangled state of $n$ qubits (i.e. an ensemble of $n$ two-level systems) as the states such that all their partial traces are maximally mixed.
The idea of characterizing \emph{multipartite maximally entangled states} (MMES) minimizing their average purity over different bipartitions of the system has been put forward in~\cite{mmes} where these states have been obtained as solutions of an optimization problem where the cost function is a proper average of purities, the \emph{potential of multipartite entanglement}:
\begin{equation*}
\pi_{\mathrm{ME}}
= \frac{1}{\mathcal{N}_A}\sum_{A}\pi_A.
\end{equation*}
Here $\mathcal{N}_A$ denotes the number of terms in the summation, which can be restricted to a certain subset of partitions  (in this paper we will consider the number of balanced bipartitions, see Definition~\ref{def:pime}).
It is interesting to notice that these states have been analyzed in different contexts. For instance, studies have been devoted to their connections with quantum secret sharing~\cite{latorre1} and  combinatorial designs~\cite{latorre2}. Moreover, recent analysis have focussed the attention on the so-called $k$-uniform states and their link to orthogonal arrays~\cite{zyc2,zyc1}. 

As already mentioned, beside an interesting topic \emph{per se}, the study of MMES is important because of new intriguing phenomena arising in the multipartite scenario. A peculiar property of multipartite entanglement, the so-called \emph{entanglement frustration}~\cite{frust} naturally appears when one tries to minimize the purity of all possible bipartitions at the same time. This subject has been explored because of its connection with self-dual codes~\cite{scott} and it has been possible to prove theorems that ensures the \emph{impossibility} to reach the ideal minimum value of purity for all bipartitions for collections of $n\ge 7$ qubits \cite{scott,siewert} and even in the relatively simple case of $n=4$~\cite{mw,gourwallach}. 

A possible approach to tackle this problem has been introduced for qubits in~\cite{cum1,Paolo,cumulants}. In particular, this approach is based on methods from classical statistical mechanics. One introduces a Hamiltonian representing the potential of multipartite entanglement
\begin{equation*}
H(z)=\pi_{ME}(z)=  \binom{n}{\left[\frac{n}{2}\right]}^{-1}\sum_{|A|=n_A/2}\pi_A=\sum_{k,k',l,l'\in\Z_2^n}\D(k,k';l,l')z_kz_{k'}\bar{z}_l\bar{z}_{l'}.
\end{equation*}
for a normalized pure state written in the computational basis in terms of its Fourier coefficients $z=(z_k)$
\begin{equation*}
\ket{\psi}=\sum_{k\in\mathbb{Z}_2^n} z_k\ket{k},
\end{equation*}
with \emph{coupling function} $\D$ (see Theorem~\ref{thm:delta} for its complete general expression).
By introducing the partition function 
\begin{equation*}
Z(\beta)=\int \d \mu(z)\e^{-\beta H(z)},
\end{equation*}
with $\beta$ a Lagrange multiplier and $\mu$ the unitarily invariant measure over pure states on the hypersphere $\{z\in\mathbb{C}^N|\|z\|^2 = \sum_k{|z_k|}^2=1\}$ induced by the Haar measure over the unitary group $\mathcal{U}(N)$~\cite{measure}:
\begin{equation*}
\d\mu(z)=\frac{(N-1)!}{\pi^N}\ \delta\left(1-\|z\|^2\right)\prod_k\d z_k\d \bar{z}_k,
\end{equation*}
one can explore the configurations for $\beta\rightarrow+\infty$ where frustration appears and only MMES are sampled. For qubits it has been possible to use an high-temperature expansion techniques and a diagrammatic evaluation of cumulants of a probability density function. In principle, one should try to perform the re-summation of all diagrams. On the other hand, it is interesting to analyze the different contribution and, using a well motivated criterion, choose a sub-class of diagrams. Obviously the choice should be a possible approach for characterizing the properties of the entanglement. Unfortunately, the calculations are far from being simple and only a few number of cumulants are analytically known. In particular, the topology of diagrams is highly non-trivial and both analytical and numerical hints suggest that the presence of frustration could be related to a precise class of graphs appearing in the cumulant expansion. One would like to find an objective procedure, if admissible, for discarding some graphs and resumming only a subset of them. Obviously, we would like to have a criterion for choosing which diagrams to maintain not only based on simplicity. A possible way to circumvent this computational difficulty has been introduced in~\cite{cactus} where the selection of graphs has been based on introducing a \emph{color} index $N_c$ and considering a field theory for the multipartite entanglement. An explicit calculation at leading order in $N_c$ has given hints about the presence of a phase transition and it has been possible to numerically observe that the limit of large values of the parameter $N_c$ removes the frustration. On the other hand, it is difficult to give a direct physical interpretation to this approach, though appealing from the mathematical point of view.

Following these motivations in this paper we want to explore another limit. In particular, after introducing a generalization of the previously sketched framework to a collection of $d$-level systems with $d>2$, we want to study and characterize the behavior of the potential of multipartite entanglement in the limit of large values of $d$. In particular, we will find an explicit expression of the coupling function $\Delta$, which generalizes the one obtained for qubits, and study its symmetries. Then we will examine the high-temperature expansion of the distribution function of the potential of multipartite entanglement, proving that the series characterizing the expansion converges and observe that when $d$ is large enough only a specific class of perturbative diagram gives a contribution to the partition function. 

The paper is organized as follows. In section~\ref{sec:multipartiteentanglement} we introduce the notation and give a detailed description of the problem. In Section~\ref{sec:main} we define and analyze the function we want to minimize, introduce the statistical mechanics approach and give the main results of the paper. In Section~\ref{sec:diagrammatic} we introduce the diagrammatic technique used for the analysis of the cumulants. Finally, using this diagrammatics, in Section~\ref{sec:hightemp} we give the proof of Theorem~\ref{thm:cactus}. We add two appendices. In Appendix~\ref{app:sstate} we exhibit numerical results about a state that, to the best of our knowledge, reaches the lowest value of the $7$-qubit potential of multipartite entanglement. In Appendix~\ref{app:codth} we include for self-consistency some results about the relation between perfect MMES and maximum distance separable codes.

\section{Bipartite and Multipartite Entanglement}
\label{sec:multipartiteentanglement}

\subsection{Bipartite Entanglement and Purity}
Let us consider a collection of $n$ $d$-dimensional quantum systems described by an $N$-dimensional Hilbert space 
$\mathcal{H} = \mathbb{C}^{d^n}$ (with $N=d^n$) and separate them into two disjoint sets of, respectively, $n_A$ and $n_{\bar{A}}$ elementary systems ($n_A+n_{\bar{A}}=n$), thus defining a bipartition.
\begin{definition}
A \textit{bipartition} of a system $S=\{1,2,\dots,n\}$ of $n$ parties is a pair $(A,\bar{A})$ such that $A\cup \bar{A}=S$ and $A\cap \bar{A}=\phi$. Furthermore, if $|A|=n_A$ and $|\bar{A}|=n_{\bar{A}}=n-n_A$ are the dimensions of the two subsystems then the bipartition is called \textit{balanced} if 
\begin{equation*}
n_A=\left[\frac n2\right]\quad \mathrm{and} \quad n_{\bar{A}}=\left[\frac{n+1}{2}\right],
\end{equation*}
where $[x]$ denotes the integer part of $x$ (greatest integer less than $x$).
\end{definition}
Notice that in the definition  it is assumed, without loss of generality, that $n_A\leq n_{\bar{A}}$; indeed, the bipartitions $(A,\bar{A})$ and $(\bar{A}, A)$ will play the same role in our considerations.  

With this definition we can consider the Hilbert space $\mathcal{H}$ as a tensor product $\mathcal{H}=\mathcal{H}_A \otimes \mathcal{H}_{\bar{A}}$ where $\mathcal{H}_A \simeq  \mathbb{C}^{N_A}$, 
$\mathcal{H}_{\bar{A}}  \simeq  \mathbb{C}^{N_{\bar{A}}}$ with dimensions, respectively, $N_A=d^{n_A}$ and $N_{\bar{A}}=d^{n_{\bar{A}}}$. Every normalized vector $|\psi\rangle \in \mathcal{H}$, representing a pure state of the system, admits a Fourier expansion in terms of the orthonormal computational basis  $\{\ket{k}\}_{k\in\mathbb{Z}_d^n}$ 
\begin{equation*}
\ket{\psi}=\sum_{k\in\mathbb{Z}_d^n} z_k\ket{k},
\end{equation*}
where $z_k=\braket{k}{\psi}\in \C$, $k\in\mathbb{Z}_d^n$ and $\mathbb{Z}_d=\mathbb{Z}/d\, \mathbb{Z}$ is the cyclic group with $d$ elements. Indeed, there is a natural correspondence between the basis of the space and the strings of length $n$ over $d$ symbols.

A convenient measure of bipartite entanglement is the so-called \textit{purity (of the reduced state)} :
\begin{equation*}
\pi_{A}(\rho)=\tr(\rho_A^2)=\tr({\rho^2_{\bar{A}}})=\sum_k \lambda_k^2,
\end{equation*} 
where $\rho_A = \tr_{\bar{A}} \rho$ is the reduced density matrix of the subsystem $A$, with $\tr_{\bar{A}}$ denoting the partial trace over subsystem $\bar{A}$,
and $\lambda_k$'s are the eigenvalues of $\rho_A$. Since for the rest of this work we will consider only the purity of the relevant subsystem, we will refer to this quantity simply as purity.

Using this expansion, we can rewrite the purity as: 
\begin{eqnarray}\nonumber
\pi_A(\psi)&=&\tr\left((\tr_{\bar{A}}(\ketbra{\psi}{\psi})^2\right)\\ \nonumber
&=&\tr\bigg(\Big(\sum_{k,l\in\mathbb{Z}_d^n}z_k\bar{z}_l\delta_{k_A l_A}\ketbra{k_A}{l_A}\Big)^2\bigg)\\ \nonumber
&=&\tr\bigg(\sum_{k,k',l,l'\in\mathbb{Z}_d^n}z_k z_{k'} \bar{z}_l \bar{z}_{l'} \delta_{k_{\bar{A}}l_{\bar{A}}}\delta_{k'_{\bar{A}}l'_{\bar{A}}}\ket{k_A}\braket{l_A}{k'_A}\ket{l'_A}\bigg)\\
&=&\sum_{k,k',l,l'\in\mathbb{Z}_d^n}z_k z_{k'} \bar{z}_l \bar{z}_{l'} \delta_{k_{\bar{A}}l_{\bar{A}}}\delta_{k'_{\bar{A}}l'_{\bar{A}}}\delta_{k'_A l_A}\delta_{k_A l'_A},
\label{eq:purity}
\end{eqnarray}
where we used the symbol $k_{A}$ to indicate the substring of $k$ belonging to $A$, and where $\delta_{l m}$ is the Kronecker delta.

In the following we will need this lemma, of immediate proof:
\begin{lemma}
\label{prop:bound}
Given a state $\ket{\psi}\in\mathcal{H}$ of $n$ qudits and a bipartition $(A,\bar{A})$, the following holds
\begin{enumerate}
\item $\pi_{A}(\psi)=\pi_{\bar{A}}(\psi)$,
\item $1/N_A\le\pi_{A}(\psi)\le 1$,
\end{enumerate}
where $N_A= d^{n_A}$, with $n_A=|A|$. Moreover the upper bound is reached only by separable states. 
\end{lemma}

\subsection{Multipartite Entanglement and Potential of Multipartite Entanglement}
\label{sec:pime}

In this section we will deeply extend the use of purity for the characterization of the \emph{multipartite} entanglement of a generic system of $n$ parties. 

In general, a system composed by $n>2$ parties has a number of different bipartitions that scales as $\mathcal{N}_A = O(2^n)$.
The states that saturate the lower bound of the purity for some bipartitions are called \textit{maximally entangled} with respect to these bipartitions. The Bell states are examples of maximally entangled states (here there is no need to specify the bipartition since it is unique in systems with only two components).
\begin{definition}
A  state $\ket{\psi}$ such that $\pi_{A}(\psi)=1/N_A$ with respect to every bipartition $(A,\bar{A})$, i.e.  maximally entangled with respect to every bipartition of the system is called a \emph{perfect multipartite maximally entangled state} (\emph{perfect MMES}).
\end{definition}
To determine if a given state $\ket{\psi}$ is a perfect MMES it is sufficient to check if it satisfies the minimization condition for all the balanced bipartitions, i.e.\ bipartitions with $n_A=|A|=[n/2]$. Indeed, if a state has a reduced density matrix of the form 
\begin{equation*}
\rho_A=\frac{\I}{N_A},
\end{equation*} 
for subsystem $A$, then 
\begin{equation*}
\rho_{B}=\frac{\I}{N_B},
\end{equation*}
for every smaller subsystem $B\subset A$. Therefore, the problem of finding perfect MMESs can be tackled by studying the average purity over all the balanced bipartitions. 
\begin{definition}
\label{def:pime}
The average purity over all possible balanced bipartitions is called \emph{potential of multipartite entanglement} and is given by
\begin{equation}
\label{eq:pime}
\pi_{\mathrm{ME}}(\psi)
= \binom{n}{\left[\frac{n}{2}\right]}^{-1}
\sum_{|A|=[n/2]}\pi_A(\psi).
\end{equation}
\end{definition}
As for the purity we can define a bound for the potential of multipartite entanglement.
\begin{proposition}
\label{prop:bound2}
The potential of multipartite entanglement has the following bounds:
\begin{equation*}
1/N_A \le\pi_{\mathrm{ME}}(\psi)\le 1,
\end{equation*}
with $N_A=d^{[n/2]}$, for every state $\ket{\psi}\in\mathcal{H}$.
\end{proposition}
The proof of this proposition is a  straightforward consequence of Lemma~\ref{prop:bound}.

In the case $d=2$ (qubits) it is possible to obtain an explicit expression of $\pi_{ME}$~\cite{Paolo}. We will extend this result for states of qudits.
Let us recall that, using the Fourier expansion of the state, we can write the purity of a pure state, with respect to a given bipartition as~\eqref{eq:purity}:
\begin{equation*}
\pi_A(\psi)=\sum_{k,k',l,l'\in\mathbb{Z}_d^n}z_k z_{k'} \bar{z}_l \bar{z}_{l'} \delta_{k'_A l_A}\delta_{k_A l'_A} \delta_{k_{\bar{A}}l_{\bar{A}}}\delta_{k'_{\bar{A}}l'_{\bar{A}}},
\end{equation*}
If we average the purity over all the possible bipartitions with fixed dimension, we can define a \textit{coupling function}~\cite{Paolo}:
\begin{equation}
\D(k,k';l,l';n_A)=\frac{1}{2}\tilde{\D}(k,k';l,l';n_A)+\frac{1}{2}\tilde{\D}(k',k;l,l';n_{\bar{A}}),
\label{eq:couplingfuncdef}
\end{equation}
where
\begin{equation*}
\tilde{\D}(k,k';l,l';n_A)=\binom{n}{n_A}^{-1}\sum_{|A|=n_A}\delta_{k'_A l_A}\delta_{k_A l'_A}\delta_{k_{\bar{A}}l_{\bar{A}}}\delta_{k'_{\bar{A}}l'_{\bar{A}}}.
\end{equation*}
\begin{definition}
The \textit{Hamming weight} of a string $k$ over an alphabet $\Sigma$, indicated by $|k|$, is the number of symbols that are different from the zero-symbol of the alphabet used.
\end{definition}
If the alphabet is binary, i.e. it is $\Z_2=\{0,1\}$, then the Hamming weight is nothing but the number of $1$ in the string. 

We can use the definition of Hamming weight to rewrite the coupling function in a more convenient form. 
\begin{theorem}
\label{thm:delta}
The coupling function $\D$ has the following expression:
\begin{equation*}
\D(k,k';l,l';n_A)=\delta_{k+k',l+l'}\ \delta_{0,(k-l)\wedge (k'-l)}\ f(k-l, k'-l,n_A),
\end{equation*}
where 
\begin{equation}
\label{eq:f}
f(k,l,n_A)=\frac{1}{2}\binom{n}{n_A}^{-1}\left[\binom{n-|k|-|l|}{n_A-|k|}+\binom{n-|k|-|l|}{n_A-|l|}\right],
\end{equation}
and 
\begin{equation*}
k\pm l=(k_j\pm l_j)_{j} \quad \mathrm{and} \quad k\wedge l=(\min\{k_j, l_j\})_j.
\end{equation*}
\end{theorem}
\begin{remark}
\label{rmk:binomial}
The binomial coefficients in Eq.~(\ref{eq:f}) is intended to be zero if one of its arguments is negative. The sum and difference are in $\mathbb{Z}_d$, and the minimum in the definition of $\wedge$ is taken on the (unique) representatives of $k_j$ and $l_j$ belonging to $\{0,1,2, \dots, d-1\}$.
\end{remark}
\begin{proof} 
The coupling function
\begin{equation}
\label{eq:ex_delta}
\D(k,k';l,l')=\binom{n}{n_A}^{-1}\sum_{|A|=n_A}\frac{1}{2}(\delta_{k'_A l_A}\delta_{k_A l'_A}\delta_{k_{\bar{A}}l_{\bar{A}}}\delta_{k'_{\bar{A}}l'_{\bar{A}}}+\delta_{k'_A l'_A}\delta_{k_A l_A}\delta_{k'_{\bar{A}}l_{\bar{A}}}\delta_{k_{\bar{A}}l'_{\bar{A}}})
\end{equation}
is non zero if and only if for some subset $A$ of $S=\{1,2,\dots,n\}$, with $|A|=n_A$, we have
 \begin{equation*}
k_A=l'_A, \quad k'_A=l_A, \quad k_{\bar{A}}=l_{\bar{A}}, \quad k'_{\bar{A}}=l'_{\bar{A}}.
\end{equation*} 
This imposes that if $j\in A$, $i\in \bar{A}\ $
\begin{equation*}
k_j=l'_j, \quad k'_j=l_j, \quad k_{i}=l_{i}, \quad k'_{i}=l'_{i},
\end{equation*} 
or equivalently that for $j\in S\ $ 
\begin{equation*}
k_j-l'_j=0\quad \mathrm{ and } \quad k'_j-l_j=0,
\end{equation*}
or 
\begin{equation*}
k_j-l_j=0\quad \mathrm{ and } \quad k'_j-l'_j=0.
\end{equation*}
Putting these conditions together we have that $\D\neq 0$ if and only if 
\begin{equation}
\label{eq:con_law}
k+k'=l+l'\ \mbox{ and }\ (k-l)\wedge(k'-l)=0.
\end{equation}

It remains to count the number of bipartitions $(A,\bar{A})$ that contribute to the sum in Eq.~(\ref{eq:ex_delta}). 
For this aim let us call 
\begin{eqnarray*}
S_0&=&\{i\in S\; |\; k_i=l_i=k'_i=l'_i\},\\ 
S_1&=&\{i\in S\; |\; k_i\neq l_i\ \mathrm{ or }\ k'_i\neq l'_i\}\\  
S_2&=&\{i\in S\; |\; k_i\neq l'_i\ \mathrm{ or }\ k'_i\neq l_i\}.
\end{eqnarray*}
From the previous discussion it is easy to see that $S_1\cap S_2=\phi$ and that $S=S_0+S_1+S_2$. With this new notation we can characterize a bipartition $(A,\bar{A})$ for which the contribution of the first term in the sum is non-zero , i.e.
\begin{equation*}
\delta_{k'_A l_A}\delta_{k_A l'_A}\delta_{k_{\bar{A}}l_{\bar{A}}}\delta_{k'_{\bar{A}}l'_{\bar{A}}}\neq 0,
\end{equation*}
as a bipartition such that $A\subset S_1+S_0$ and $\bar{A}\subset S_2+S_0$. Furthermore, since $A\cap \bar{A}=\phi$ and $A\cup \bar{A}=S$ then $A=S_1+A\cap S_0$ and $\bar{A}=S_2+\bar{A}\cap S_0$, we can conclude that their number is equal to the binomial coefficient
\begin{equation*}
\binom{|S_0|}{|A-S_1|}=\binom{n-|S_1|-|S_2|}{n_A-|S_1|}=\binom{n-|k-l|-|k'-l|}{n_A-|k-l|}.
\end{equation*}
The same result can be obtained for the second term in the sum:
\begin{equation*}
\delta_{k'_A l'_A}\delta_{k_A l_A}\delta_{k'_{\bar{A}}l_{\bar{A}}}\delta_{k_{\bar{A}}l'_{\bar{A}}}
\end{equation*} 
swapping the role of $A$ and $\bar{A}$, and this ends the proof.
\qed
\end{proof}
Since we are going to focus on balanced bipartitions, from now on we will omit the dependence on $n_A$ both in $\D$ and $f$, with the understanding that $n_A=\left[\frac{n}{2}\right]$. In this way, with the use of the coupling function, the potential of multipartite entanglement can be written as
\begin{equation*}
\pi_{ME}(\psi)=\sum_{k,k',l,l'\in\Z_d^n}\D(k,k';l,l')z_kz_{k'}\bar{z}_l\bar{z}_{l'}.
\end{equation*}

\subsection{MMES, Perfect MMES and Frustration}

We will see that the lower bound $1/N_A = 1/d^{[n/2]}$ of the potential of multipartite entanglement is not always attained. This justifies the following
\begin{definition} A state $\ket{\varphi}$ that minimizes $\pi_{\mathrm{ME}}$, i. e. $\pi^{\min}_{\mathrm{ME}}=\pi_{\mathrm{ME}}(\varphi)$, where 
\begin{equation*}
\pi^{\min}_{\mathrm{ME}}=\min \{\pi_{\mathrm{ME}}(\psi)\; : \; \ket{\psi}\in \mathcal{H}, \bra{\psi}\psi\rangle=1\},
\end{equation*}
is a \textit{multipartite maximally entangled state} (\textit{MMES}).
\end{definition}
Let us stress, once again, that the difference between a MMES and a perfect MMES lies in the saturation of the lower bound of the potential of multipartite entanglement.
\begin{example}
The qubit GHZ state, i.e. the state $\ket{GHZ}=\frac {1}{\sqrt{2}} (\ket{000}+\ket{111})$,  is a perfect MMES, indeed it is easy to show that the purities with respect to all the possible bipartitions are $\frac{1}{2}$.
\end{example}

One of the questions that arises naturally from the previous discussion is on the general structure of a perfect MMES for given values of $d$ (the dimension of each subsystem) and $n$ (the number of subsystems). 

With an abuse of notation we can say that the Bell states are perfect MMES for systems of two qubits\footnote{There is no multipartite entanglement here.} while for $n=3$ qubits the only perfect MMES, up to local and unitary transformations, is the GHZ state. 

The problem of characterizing a perfect MMES has not always such an easy solution. In~\cite{gourwallach} Gour {\it et al.} proved that for $n=4$ qubits a perfect MMES does not exist and that the minimum value the average purity can attain  is 
\begin{equation*}
\pi_{ME}^{\min}=\frac{1}{3}>\frac{1}{4}=\frac{1}{N_A}.
\end{equation*} 

When the lower bound of the potential of multipartite entanglement cannot be saturated, the system is said to be \textit{frustrated}. If this is the case the requirement that the purity be minimal for all the bipartitions generates conflicts among them. 

For system of $n=5, 6$ qubits there are examples of perfect MMES, see~\cite{Paolo}, while for $n\geq 8$ qubits a perfect MMES does not exist as proved by Scott in~\cite{scott}, using  classical error correction theory. 
The case of $n=7$ qubits has been recently shown to be frustrated \cite{siewert}. On the other hand, the value of $\pi_{ME}^{\min}$ in this case is unknown and so the structure of the associated MMES. Up to now only numerical estimates about the minimum of the potential of multipartite entanglement have been done. For a lower bound of  $\pi_{ME}^{\min}$ for $7$-qubits, see appendix~\ref{app:sstate}. 

Frustration appears when one or more bipartitions cannot reach their minima. Nevertheless, it can be proven that enlarging the dimension $d$ of each subsystem, at fixed $n$, tends to eliminate this problem, and in particular that there exist values of $d\geq n+1$ for which it is possible to find a perfect MMES of $n$ qudits. For a discussion on this statement see appendix~\ref{app:codth}.

\section{Main Results}
\label{sec:main}
In this section we want to go briefly throughout the main results of this paper, before going  into the details of the proofs. 

\subsection{Simmetries of the coupling function $\D$}
\label{deltasymmetries}

We recall that,~\eqref{eq:ex_delta}
\begin{equation*}
\D(k,k';l,l')=\binom{n}{n_A}^{-1}\sum_{|A|=n_A}\frac{1}{2}(\delta_{k'_A l_A}\delta_{k_A l'_A}\delta_{k_{\bar{A}}l_{\bar{A}}}\delta_{k'_{\bar{A}}l'_{\bar{A}}}+\delta_{k'_A l'_A}\delta_{k_A l_A}\delta_{k'_{\bar{A}}l_{\bar{A}}}\delta_{k_{\bar{A}}l'_{\bar{A}}}).
\end{equation*}
Due to its form, it is invariant under the permutation of the qudits and under some swaps of the computational basis elements ($k\in\Z_d^n$). 

It is well known that applying local  unitary transformations to the system does not change its entanglement, and as a consequence the local purity of any of its subsystem:
\begin{equation*}
\pi_A(\psi) = \pi_A\big( (U_1\otimes U_2 \otimes \dots \otimes U_n) \psi \big),
\end{equation*}
for all $\psi \in \mathcal{H}$, for  all $A\subset S$ and for all $(U_1, \dots, U_n) \in \mathrm{U}(d)^n$, with $\mathrm{U}(d)$ being the unitary group of degree $d$.
Moreover, if we permute the order of the qudits the global entanglement of the system is left invariant. The permutation group $\mathrm{S}_n$ of order $n$ acts on $\mathcal{H}$ through (unitary) swap operators $p\in\mathrm{S}_n \to V_p \in \mathrm{U}(N)$:
\begin{equation*}
V_p \big(|\psi_1\rangle\otimes |\psi_2\rangle \otimes \dots\otimes |\psi_n\rangle\big) = |\psi_{p(1)}\rangle\otimes |\psi_{p(2)}\rangle \otimes \dots\otimes |\psi_{p(n)}\rangle.
\end{equation*}
For all $(U_1, \dots, U_n) \in \mathrm{U}(d)^n$ and for all $p\in\mathrm{S}_n$ we get that  
\begin{equation*}
\big((U_1, \dots, U_n; p)\pi_{\mathrm{ME}}\big)(\psi):=  \pi_{\mathrm{ME}} \big( (U_1\otimes U_2 \otimes \dots \otimes U_n) V_p \psi \big) = \pi_{\mathrm{ME}}(\psi).
\end{equation*}
Therefore, the potential of multipartite entanglement~(\ref{eq:pime}) admits the semidirect product
\begin{equation*}
\mathrm{SU}(d)^n\rtimes \mathrm{S}_n
\end{equation*}
as symmetry group, whose product is easily seen to satisfy
\begin{equation*}
 (U_1, \dots, U_n; p) (V_1, \dots, V_n; q) = (U_1 V_{p(1)}, \dots, U_n V_{p(n)} ; p \, q).
\end{equation*}

As a consequence, for the symmetries of the coupling function $\D$, we have the following:
\begin{theorem}
\label{thm:symmetry}
The coupling function $\D$ in~\eqref{eq:couplingfuncdef} is invariant under the action of the semidirect product group
\begin{equation*}
\mathrm{S}_d^n\rtimes \mathrm{S}_n,
\end{equation*}
whose action on $k\in\Z_d^n$ is given by
\begin{equation*}
(p_1,\dots, p_n;q)(k_1,\dots,k_n)=\big(p_1(k_{q(1)}),\dots,p_n(k_{q(n)})\big),
\end{equation*}
where $p_j\in \mathrm{S}_d$, $\forall j\in \{1,\dots, n\}$, and $q\in \mathrm{S}_n$. 
\end{theorem}
\begin{proof}
The proof is straightforward after observing that all the operations that characterize the coupling function  act position-wise and the permutations are bijective maps.  
\qed
\end{proof}

\subsection{Statistical mechanics approach and cumulant expansion}

The minimization problem of the potential of multipartite entanglement can be handled following a statistical mechanics approach~\cite{cum1}. Roughly speaking we will consider the free energy of a suitable classical system at a fictitious temperature and we will recover the original problem in the zero temperature limit.

Considered the state
\begin{equation*}
\ket{\psi}=\sum_{k\in\Z_d^n} z_k\ket{k},
\end{equation*} 
with $z=(z_k)_k$ the vector of the Fourier coefficients in the expansion of the state, $\|z\|^2 =\sum_k |z_k|^2 = 1$, we define the Hamiltonian
\begin{equation*}
H(z)=\pi_{ME}(\psi(z)).
\end{equation*}

Let us consider $M$ vectors and the ensemble $\{m_j\}$ of the number of vectors with fixed potential of multipartite entanglement, $H = \epsilon_j$. We want to find the distribution that maximizes the quantity 
\begin{equation*}
\Omega=\frac{M!}{\Pi_jm_j!},
\end{equation*}
under the constraints $\sum_j m_j=M$ and $\sum_j m_j \epsilon_j=ME$, where $E$ is the average value of $\pi_{ME}$. In particular, if we let $M\to \infty$ we recover the canonical ensemble with partition function
\begin{equation}
Z(\beta)=\int \d \mu(z)\e^{-\beta H(z)},
\label{eq:partition function}
\end{equation}
where 
\begin{equation*}
\d\mu(z)=\frac{(N-1)!}{\pi^N}\ \delta\left(1-\|z\|^2
\right)\prod_k\d z_k\d \bar{z}_k,
\end{equation*}
is the unitarily invariant measure over pure states induced by the Haar measure over $\mathcal{U}(\H)$ through the mapping $\ket{\psi} =U\ket{\psi_0}$, for a given state $\ket{\psi_0}$~\cite{measure}. Here $\beta$ plays the role of an inverse temperature, so that for $\beta\to+\infty$ only the configurations that minimize the Hamiltonian survive. In other words, we recover the MMES in the limit $\beta\to+\infty$. Moreover, if $\beta\to 0$, we recover the behaviour of a typical state.
 
Using the partition function, the average energy can be written as
\begin{equation}
\langle H\rangle_{\beta}=\frac{1}{Z(\beta)}\int \d \mu(z)H(z) \e^{-\beta H(z)}=-\frac{\partial}{\partial\beta}\ln Z(\beta).
\label{eq:average energy}
\end{equation} 
The high-temperature expansion of this energy distribution is:
\begin{equation}
\label{eq:series}
\langle H\rangle_\beta= \sum_{m=1}^{\infty} \frac{(-\beta)^{m-1}}{(m-1)!}\ \kappa^{(m)}_{0}[H],
\end{equation}
where
\begin{equation*}
\kappa^{(m)}_{0}[H] =(-1)^m \frac{\partial^{m}}{\partial\beta^{m}} \ln Z(\beta) \Big|_{\beta=0}
=(-1)^{m-1}\frac{\partial^{m-1}}{\partial\beta^{m-1}}\langle H\rangle_{\beta} \Big|_{\beta=0}
\end{equation*}
is the $m$-th cumulant, that are related to the  moments $\langle H^m\rangle_0$
through the recursion formula:
\begin{equation}
\label{eq:reccum}
\kappa^{(m)}_{0}=\langle H^m\rangle_0 -\sum_{j=1}^{m-1}\binom{m-1}{j-1}\kappa^{(j)}_{0}\langle H^{m-j}\rangle_0,
\end{equation}
with $\kappa^{(1)}_{0} = \langle H\rangle_0$.

This approach based on methods from classical statistical mechanics has been applied to qubits both in the bipartite~\cite{bipcum} and in the multipartite case~\cite{cum1,cumulants}. Here we want to analyze the general qudit case.

We observe that the series in Eq.~\eqref{eq:series} converges. Indeed, we can prove that it is majorized term by term by an absolutely convergent series:
\begin{theorem}
\begin{enumerate}
\item The partition function $Z(\beta)$ in~(\ref{eq:partition function}) is an entire function of $\beta\in\mathbb{C}$; 
\item The average energy~(\ref{eq:average energy}) is holomorphic in a complex neighborhood of the real line; 
\item Its high-temperature expansion~(\ref{eq:series}) is a convergent series with a nonzero radius of convergence. 
\end{enumerate}
\end{theorem}
\begin{proof} 
\begin{enumerate}
\item 
Notice that the measure in~(\ref{eq:partition function}) has  compact support $\{z\in \mathbb{C}^N, \|z\|=1\}$, and $H(z)$ is a continuous function with $1/N_A \leq H(z) \leq 1$ for $z$ in that support. Thus the integral converges for all $\beta\in\mathbb{C}$ and is differentiable with derivative given by  
$$\frac{\d Z(\beta)}{\d\beta}= \int \d \mu(z)H(z) \e^{-\beta H(z)},$$ implying that $Z(\beta)$ is holomorphic in the whole complex plane.

\item
This follows from the observation that the average energy is the ratio of two entire function $\langle H\rangle_{\beta}= Z'(\beta)/Z(\beta)$ and for $\beta \in \mathbb{R}$ and $\|z\|=1$ one gets
$$\e^{-\beta H(z)} \geq \e^{-|\beta||H(z)|} \geq \e^{-|\beta|},$$
implying that 
$$Z(\beta) \geq \e^{-|\beta|}>0, \qquad \beta \in\mathbb{R}.$$
(In fact, for $\beta<0$ one gets the stronger estimate $Z(\beta) \geq \e^{|\beta|/N_A} \geq 1$). The statement follows by continuity.

\item
This follows from statement 2. Notice in particular that $Z(0)=1$.
\qed
\end{enumerate}
\end{proof}

Finally, the following bounds hold:
\begin{theorem}
\label{thm:cactus}
\begin{enumerate}

\item
For all $m\geq1$ the moment of the Hamiltonian has the form
\begin{equation}
\label{eq:momentmajorization}
\langle H^m\rangle_0 = \langle H^m\rangle_{0,\mathrm{C}} + \langle H^m\rangle_{0,\mathrm{NC}},
\end{equation}
where
\begin{equation*}
\langle H^m\rangle_{0,\mathrm{C}} = \frac{C_1(m) N! }{(N+2m-1)!}\left(\frac{N_A + N_{\bar{A}}}{2}\right)^m
\end{equation*}
and
\begin{equation*}
0\leq \langle H^m\rangle_{0,\mathrm{NC}} \leq \frac{C_2(m) N! }{(N+2m-1)!}\left(\frac{N_A + N_{\bar{A}}}{2}\right)^{m-1},
\end{equation*}
with $C_1(m)$ and $C_2(m)$ being positive functions of the parameter $m$ only, that do not depend on $d$ or $n$.

\item The following bound holds
\begin{equation*}
0 \leq \frac{\langle H^m\rangle_{0,\mathrm{NC}}}{\langle H^m\rangle_{0,\mathrm{C}}} \leq \frac{C(m)}{d^{\left[\frac{n}{2}\right]}} ,
\end{equation*}
where $C(m)=C_2(m)/C_1(m)$.
\end{enumerate}

\end{theorem}
This is the central result of our paper. In principle, this majorization allow us to evaluate the terms in the series~(\ref{eq:series}), using eq.~\eqref{eq:reccum}. In particular, we notice that, since $C(m)$ does not depend on $d$, in the limit $d\to \infty$ the contribution of the second term, $\langle H^m\rangle_{0,\mathrm{NC}}$,  in Eq.~(\ref{eq:momentmajorization}) is subdominant and
\begin{equation*}
\langle H^m\rangle_0 \sim \langle H^m\rangle_{0,\mathrm{C}}, \qquad d\to\infty.
\end{equation*}
We will show in the following that this behavior can be interpreted in terms of the structure of graphs contributing to the moments.

We will give a proof of this theorem in Section~\ref{sec:hightemp}. Before doing this, we will introduce in Section~\ref{sec:diagrammatic} the diagrammatics used for the majorization of the moments in Eq.~(\ref{eq:momentmajorization}).

\section{Cactus and Other Diagrams}
\label{sec:diagrammatic}
In this section we will use the diagrammatic technique introduced in~\cite{cumulants} for qubits, properly generalized for the case of qudits, in order to control each term of the series~\eqref{eq:series}.

First of all let us consider the quantity:
\begin{equation*}
\langle H^m\rangle_0=\Big\langle \Big(\sum_{k,k',l,l'\in\Z_d^n}\D(k,k';l,l')z_k z_{k'} \bar{z}_l \bar{z}_{l'}\Big)^m\Big\rangle_0.
\end{equation*}
An explicit form of this quantity requires the product of $m$ coupling functions  $\D$. In order to simplify the notation we introduce the vectors 
$$\bm{k}=(k_1,\dots,k_{m}, k_{1'}, \dots, k_{m'}),
\quad \bm{l}=(l_1,\dots,l_{m}, l_{1'}, \dots, l_{m'})
$$ with $k_{j}, k_{j'},l_{j}, l_{j'}\in \mathbb{Z}_d^n$. 
Therefore,
\begin{equation*}
\langle H^m\rangle_0=\sum_{\bm{k},\bm{l}\in \mathbb{Z}_d^{2mn}}\prod_{j=1}^{m} \D(k_{j},k_{j'};l_{j},l_{j'}) \Big\langle \prod_{j=1}^{m}z_{k_j}z_{k_{j'}}\bar{z}_{l_j} \bar{z}_{l_{j'}} \Big\rangle_0.
\end{equation*}

\begin{theorem}
The following equality holds:
\begin{equation}
\label{eq:momentsimplified}
\langle H^m\rangle_0=\frac{1}{N(N+1)\dots (N+2m-1)}\sum_{\bm{k}\in  \mathbb{Z}_d^{2mn}}\sum_{p\in \mathrm{S}_{2m}} \prod_{j=1}^{m}\D(k_{j},k_{j'};k_{p(j)},k_{p(j')}),
\end{equation}
with $p\in\mathrm{S}_{2m}$ being a permutation acting on the $2m$ elements $$\{1,2,\dots,m,1',2',\dots m'\}.$$
\end{theorem}
This theorem was given in~\cite{cumulants} for qubits, $d=2$. The proof of its extension to qudits $d>2$ is a carbon copy of the proof for qubits.

By defining the square brackets
\begin{equation*}
 [p(1)\ p(1'), \dots, p(m)\ p(m')]:=\sum_{\bm{k}\in  \mathbb{Z}_d^{2mn}}\prod_{j=1}^{m} \D(k_{j},k_{j'};k_{p(j)},k_{p(j')}), 
\end{equation*}
with $p\in\mathrm{S}_{2m}$, 
Eq.~(\ref{eq:momentsimplified}) becomes
\begin{equation}
\label{eq:cumulants2}
\langle H^m\rangle_0=\frac{1}{N(N+1)\dots (N+2m-1)}\sum_{p\in \mathrm{S}_{2m}}[p(1)\ p(1'), \dots, p(m)\ p(m')].
\end{equation}
As promised we can give a diagrammatic representation of the terms in the sum.
Each pair $(k_j,k_{j'})$  can be associated to a vertex of a graph from which two edges go out and two go in. The first two edges are labeled by $k_{p(j)}$ and $k_{p(j')}$, and the latter are $k_j$ and $k_{j'}$, see fig.~\ref{fig:cactus1}.
\begin{figure}[!h]
   \centering
   \includegraphics[scale=0.20]{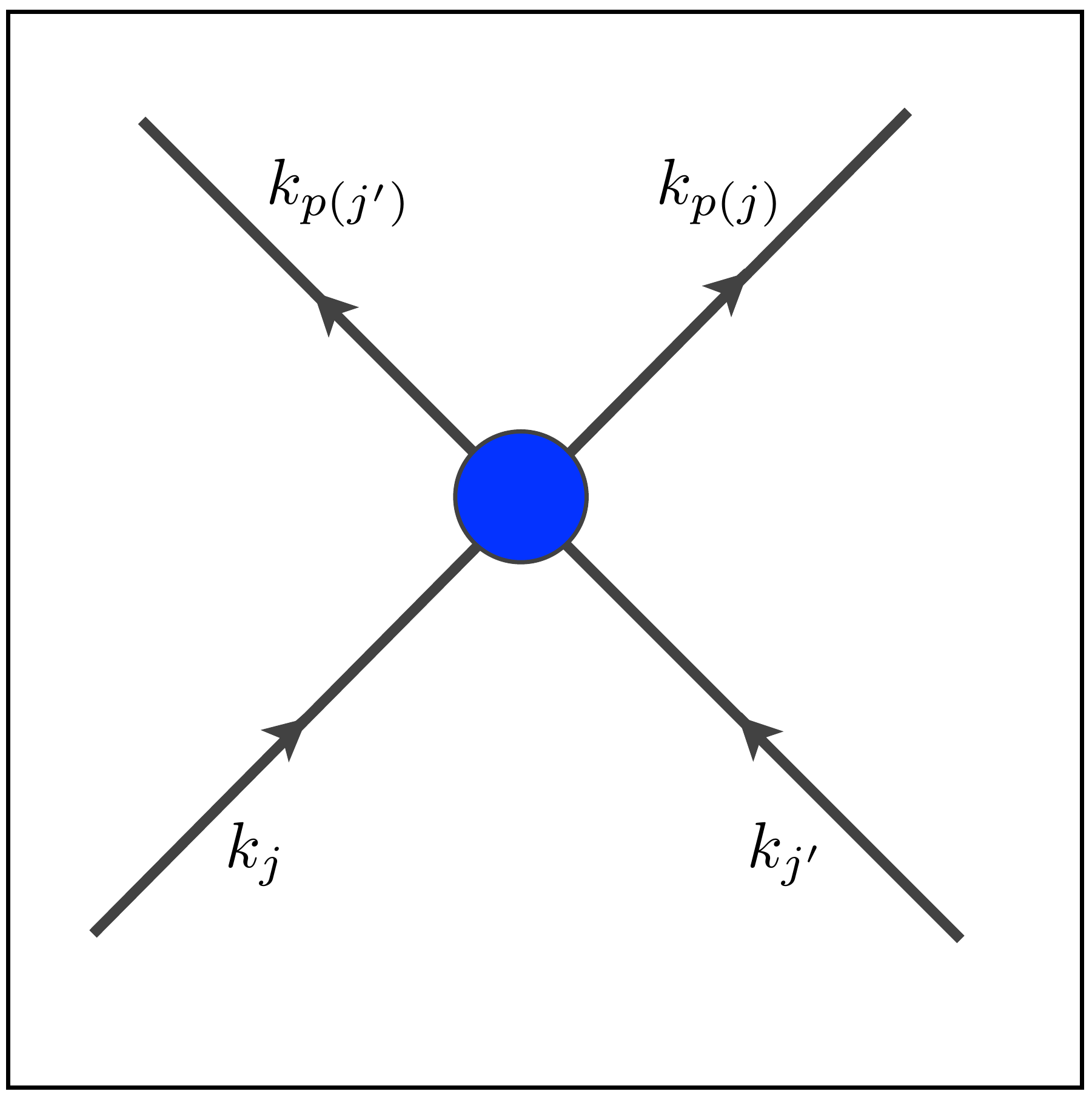}
 \caption{Graphical representation of the interaction of each pair $(k_i,k_{i'})$: each vertex with 4 edges, two going in and two going out.}
  \label{fig:cactus1}
\end{figure}
\begin{example}
\label{ex:2pt}
The square brackets $[1\ 2,1'\ 2']$ leads to the graph in fig.~\ref{fig:cactus2}.
\end{example}
\begin{figure}[!h]
  \centering
   \includegraphics[scale=0.20]{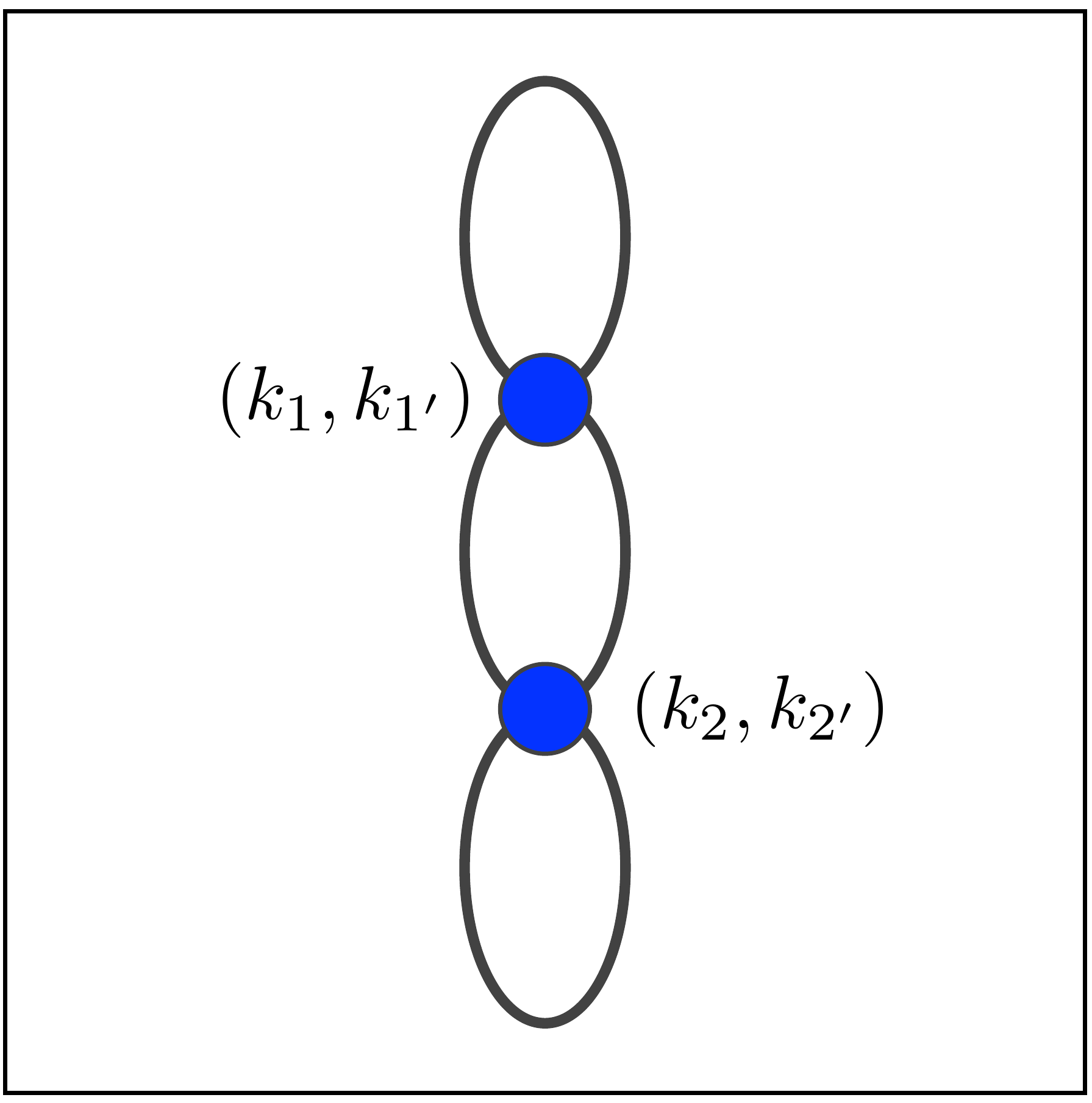}
 \caption{Graph with two points representing $[1\ 2,1'\ 2']$.}
 \label{fig:cactus2}
\end{figure}

It is possible to rephrase some of the previous results in terms of these diagrams. Indeed, Eq.~(\ref{eq:con_law}) can be interpreted as a current conservation law, i.e. the current going into a vertex has to be the same as the current that goes out, see fig.~\ref{fig:cactus1}. Moreover, the symmetries of the coupling function $\D$, given in Theorem~\ref{thm:symmetry}, are translated in a degeneracy of the graphs. For instance, the square brackets in example~\ref{ex:2pt} leads to the same graph as
\begin{equation*}
[1\ 2,1'\ 2'],\, [2\ 1,2'\ 1'],\, [1'\ 2',1\ 2],\, [2'\ 1',1\ 2],\, [1\ 2,2'\ 1']
\end{equation*}
and so on.
\begin{example}
In terms of Feynman graphs we have
\begin{equation*}
\langle H\rangle_0=\frac{1}{N(N+1)}\left([1\ 1']+[1'\ 1]\right)=\frac{1}{N(N+1)}(\ \raisebox{-6.5 ex}{\includegraphics[scale=0.20]{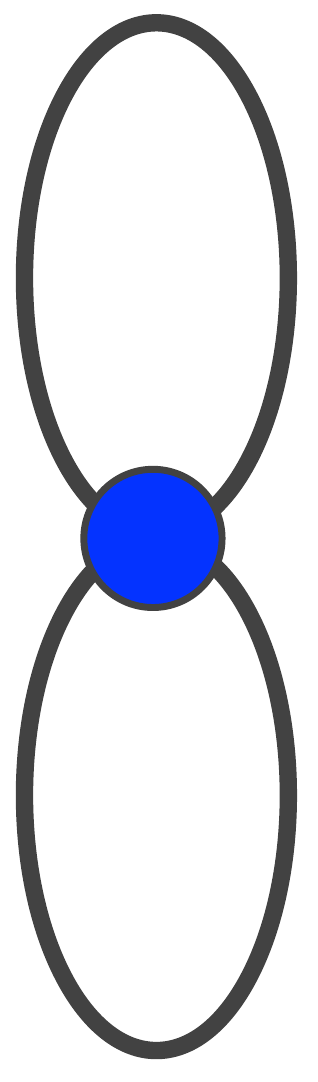}}+\raisebox{-6.5 ex}{\includegraphics[scale=0.20]{doppioloop.pdf}}\ )
\end{equation*}
\end{example}

\begin{definition}
A connected graph with $v\geq 2$ vertices is called \emph{cactus} if for every vertex  there exists a pair of edges such that removing them the graph becomes disconnected, otherwise the graph is called \emph{non-cactus}. A graph with $v=1$ is a cactus by definition.
\end{definition}
\begin{example}
The graph in fig.~\ref{fig:cactus2} is a cactus, while the graph in fig.~\ref{fig:cycle} is a non-cactus. Indeed, by removing a pair of edges from a vertex, the graph becomes one of the two graphs in fig.~\ref{fig:noncactusmonco}.
\begin{figure}[!h]
  \centering
  \includegraphics[scale=0.20]{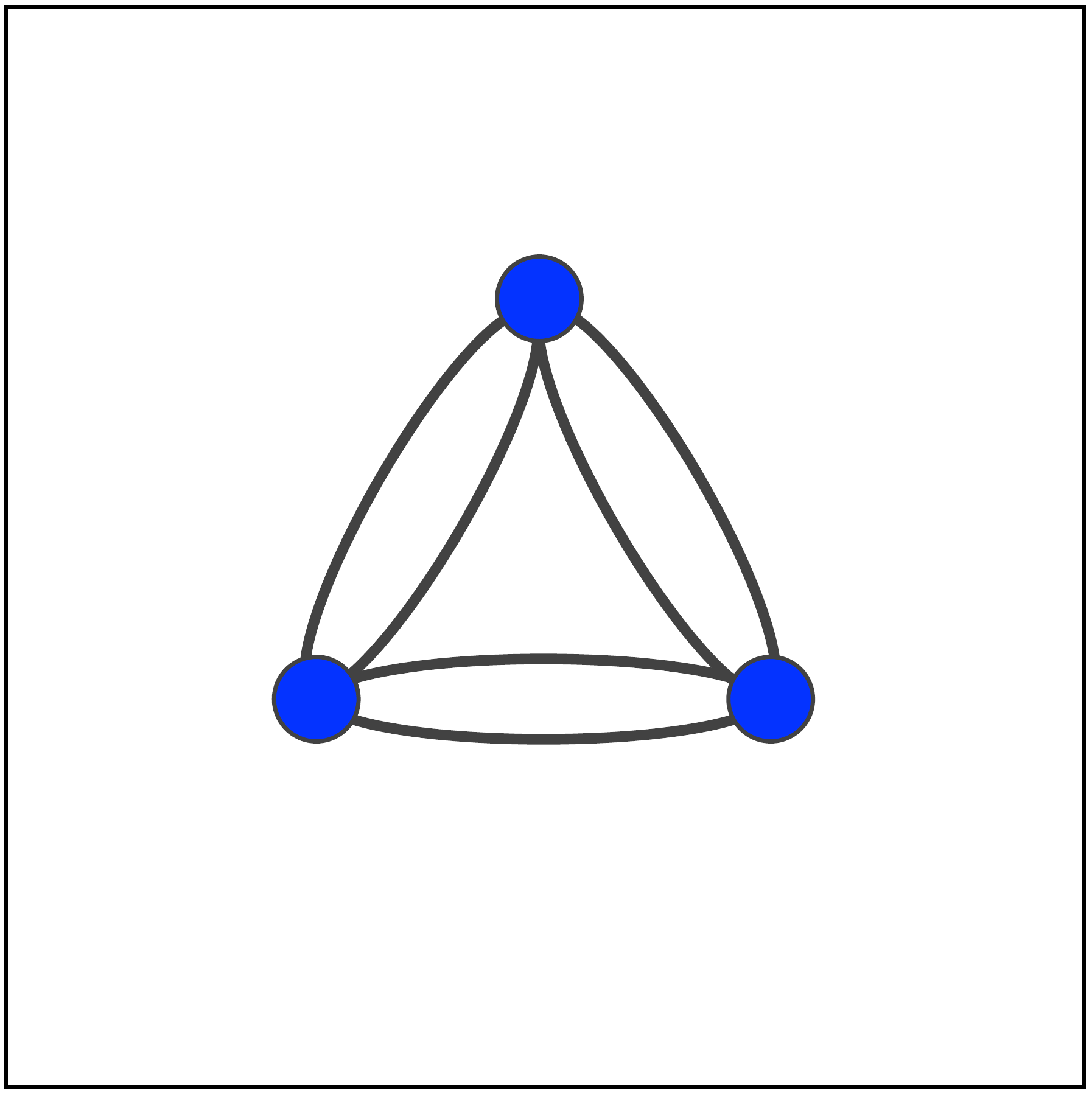}
 \caption{A graph with a cycle is a non-cactus.}
 \label{fig:cycle}
 \end{figure}
\begin{figure}[!h]
 \begin{minipage}[b]{7cm}
   \centering
   \includegraphics[scale=0.20]{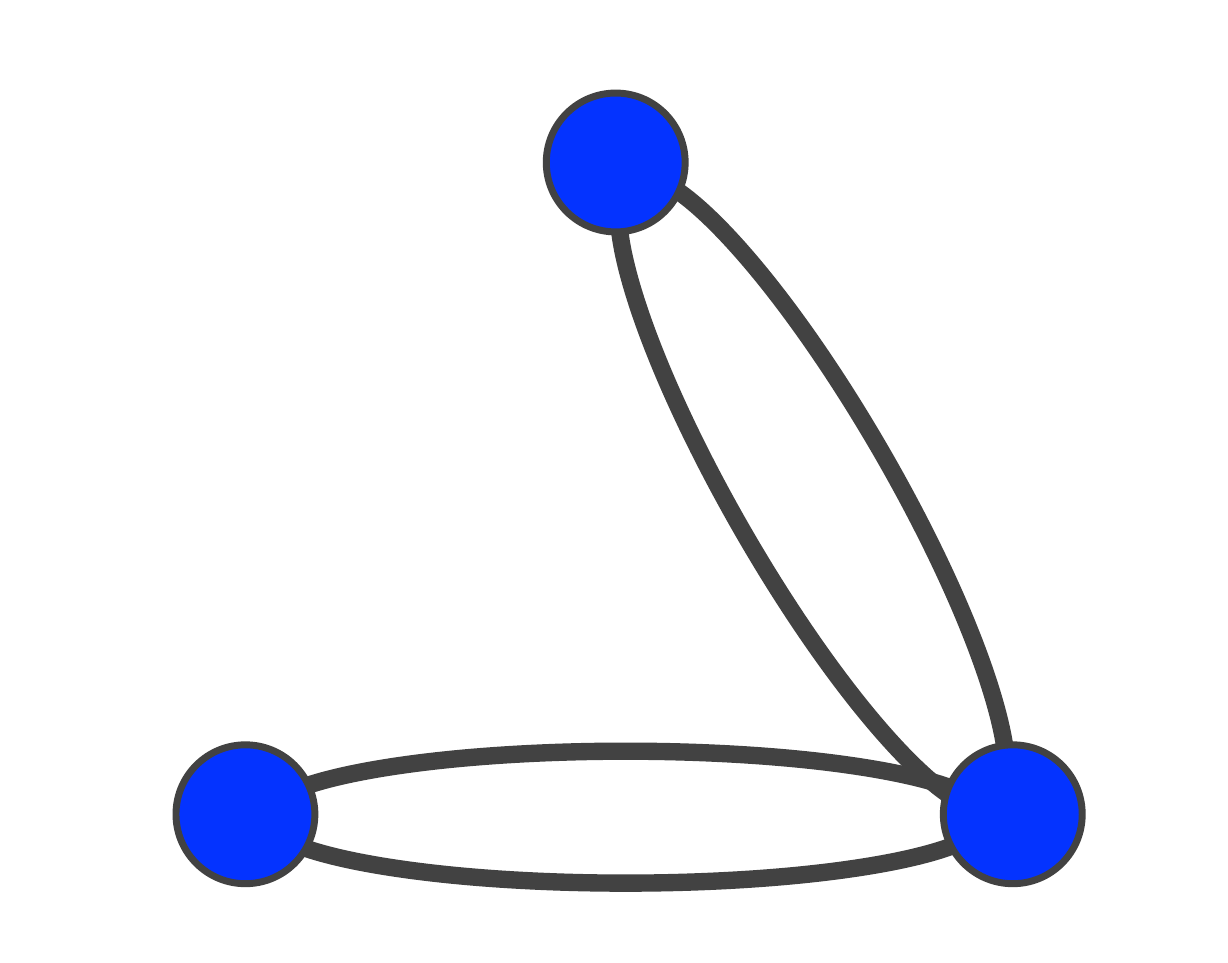}
 \end{minipage}
 \ \hspace{-15mm} 
 \begin{minipage}[b]{5cm}
  \centering
  \includegraphics[scale=0.20]{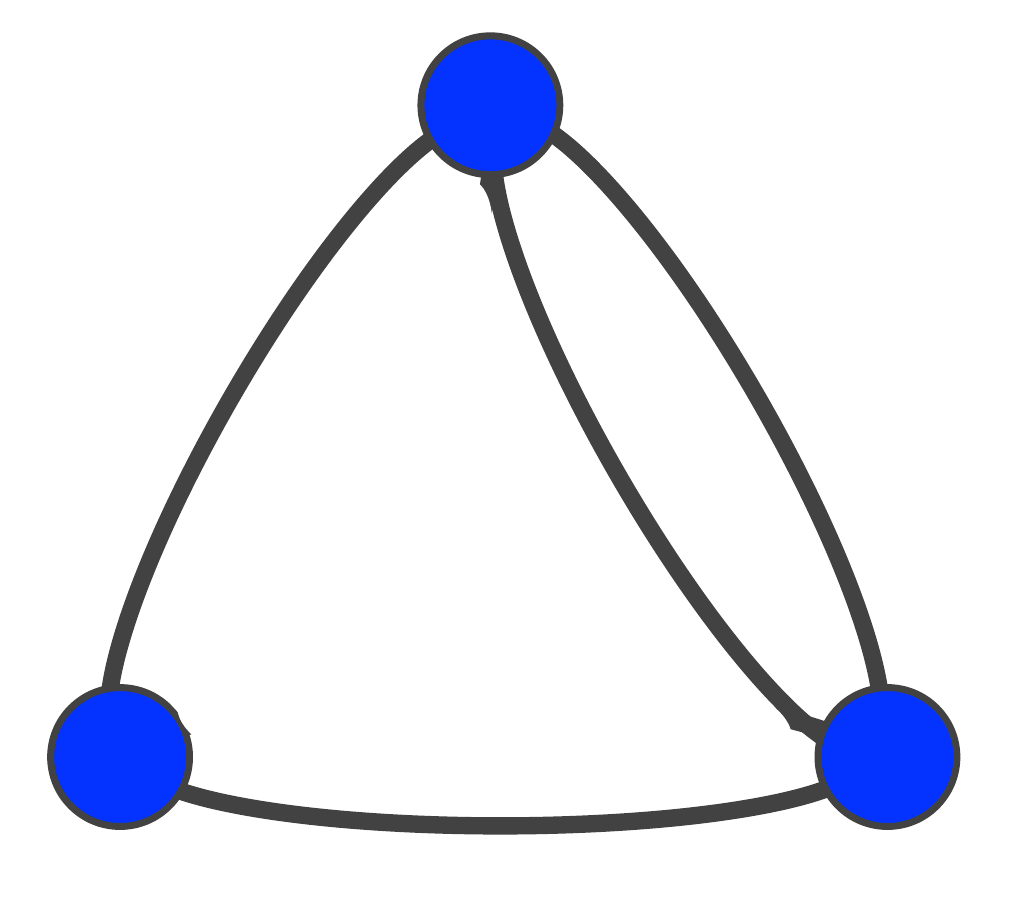}
 \end{minipage}
\caption{Removing two edges from a vertex of the non-cactus in Fig.~\ref{fig:cycle} leaves the graph connected.}
\label{fig:noncactusmonco}
\end{figure}
\end{example}
\subsection{Graph surgery}
\label{sec:deg} 
In this section we will study in detail the graphs introduced in the previous section. In particular we will compute the contribution that each graph gives to the moments~(\ref{eq:cumulants2}) and their degeneracy. In order to do this we will divide each graph in subgraphs and  will compute the degeneracy and the contribution of each single subgraph.
 
\begin{definition}
We call \textit{leaf} the subgraph of a graph represented by 
\begin{equation*}
[\dots, j\ p(j'),\dots]=\raisebox{-2.5 ex}{\includegraphics[scale=0.20]{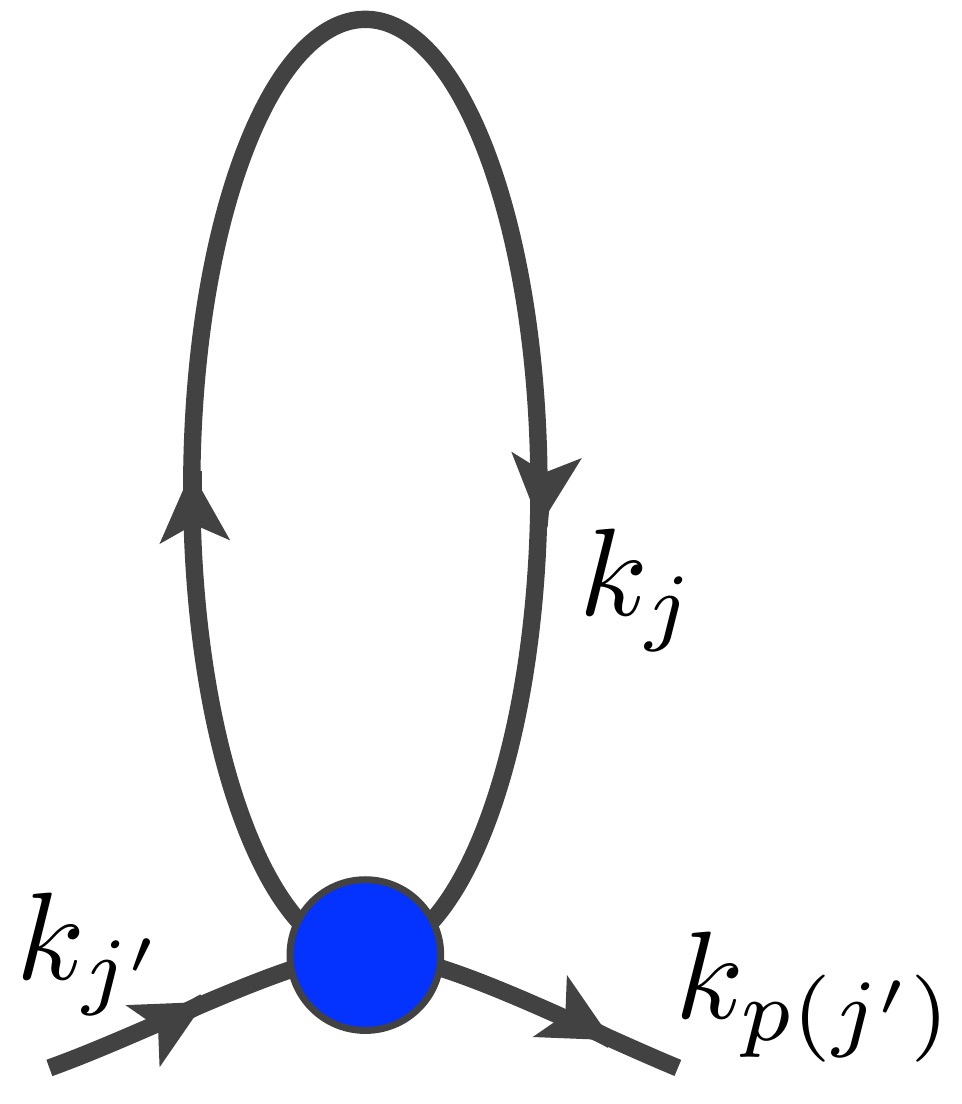}}.
\end{equation*}
\end{definition}
\begin{lemma}
\label{thm:leaf}
A leaf at the vertex $(k_j, k_{j'})$ gives the contribution:
\begin{equation*}
\delta_{k_{j'},k_{p(j')}} \frac{N_A+N_{\bar{A}}}{2}.
\end{equation*}
\end{lemma}
\begin{proof}
In a graph the contribution of a leaf is:
\begin{equation*}
[\dots, j\ p(j'),\dots]=\sum_{\substack{k_i:i\neq j\\ k_{i'}}}\dots\sum_{k_j}\D(k_j,k_{j'};k_j,k_{p(j')}).
\end{equation*}
Since this is the only term in which the index $j$ appears, it can be isolated from the rest:
\begin{eqnarray*}
\sum_{k_j}\D(k_j,k_{j'};k_j,k_{p(j')})&=&\delta_{k_{j'},k_{p(j')}}\sum_{k_j} f(k_{j'}-k_j,0)\\
&=&\delta_{k_{j'},k_{p(j')}} \frac{N_A+N_{\bar{A}}}{2}.
\end{eqnarray*}
\flushright\qed
\end{proof}
\begin{figure}[!h]
   \centering
   \includegraphics[scale=0.20]{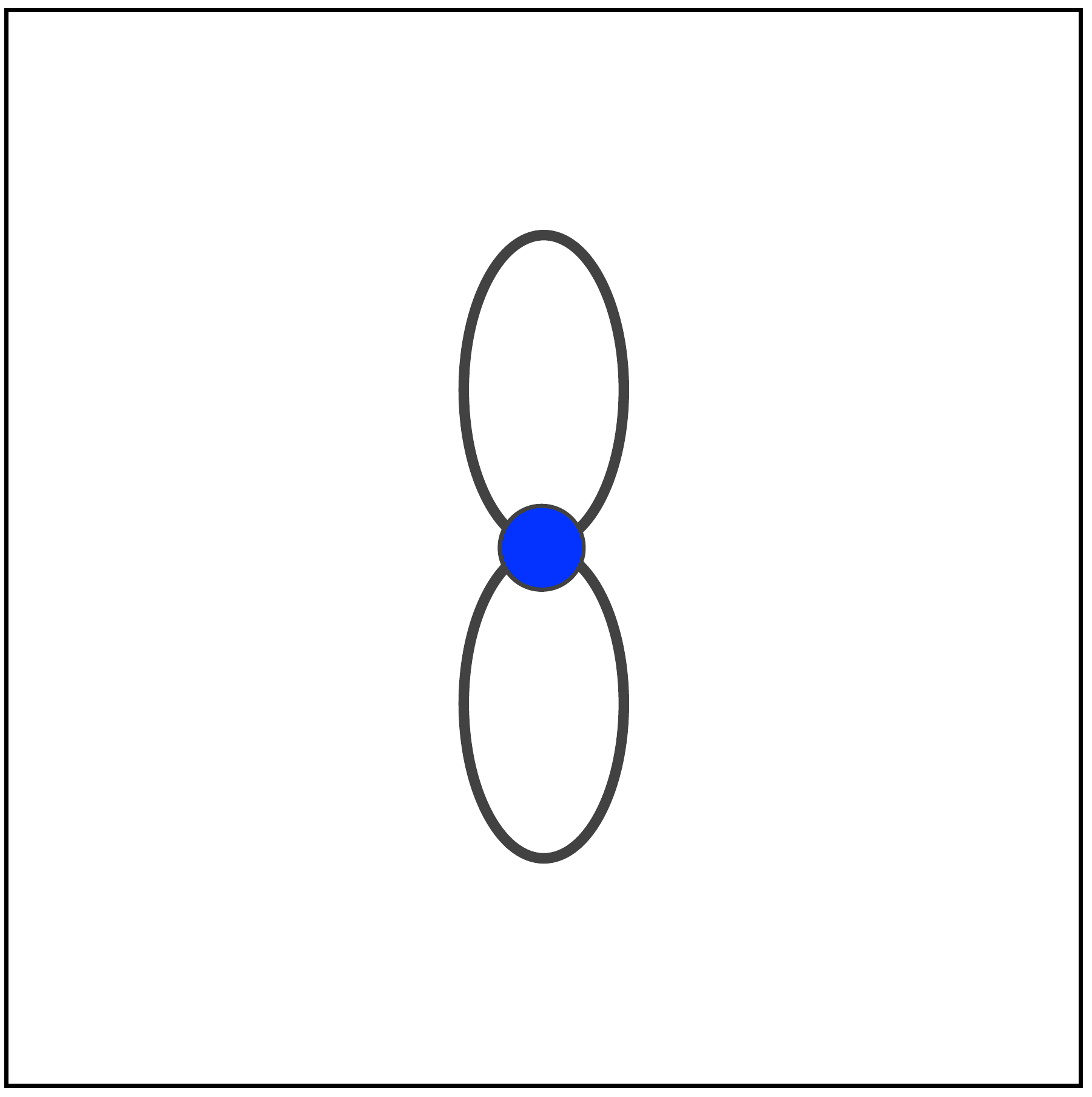}
\caption{The only graph with a single vertex.}
 \label{fig:otto}
 \end{figure}
\begin{example}
As an example consider the Feynman graph in fig.~\ref{fig:otto}. By computing the contribution of the upper leaf as in the previous lemma we find
\begin{equation*}
\raisebox{-6.5 ex}{\includegraphics[scale=0.20]{doppioloop.pdf}}=\sum_{k_{1'}}\delta_{k_{1'},k_{1'}} \frac{N_A+N_{\bar{A}}}{2}=N\; \frac{N_A+N_{\bar{A}}}{2}.
\end{equation*}
\end{example}
\begin{definition}
A \textit{loop} is the subgraph of a graph represented by 
\begin{equation*}
[\dots, p(i)\ j',\dots, p(j)\ i',\dots]=\raisebox{-6.5 ex}{\includegraphics[scale=0.20]{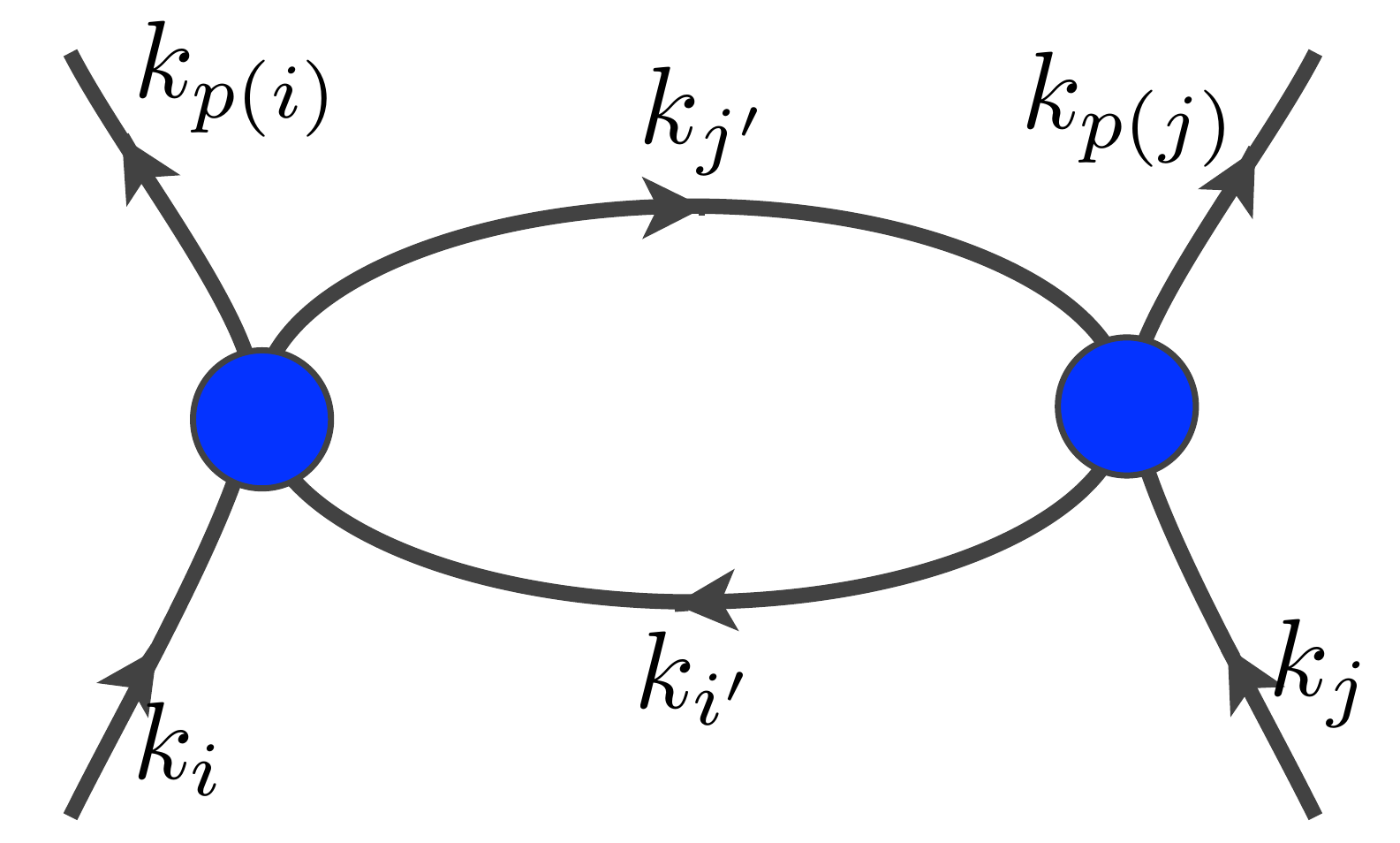}}.
\end{equation*}
\end{definition}
\begin{theorem}
\label{thm:cactus_cont}
Each cactus graph gives a contribution 
\begin{equation*}
N\left(\frac{N_A+N_{\bar{A}}}{2}\right)^v,
\end{equation*}
where $v$ is the number of vertices in the graph. 
\end{theorem}
\begin{proof}
We can compute the contribution of a graph decomposing it in its elementary subgraphs. 
From its definition, we can deduce that a cactus has at least one leaf. Moreover, notice that after we have computed the contribution of the leaf, the remaining terms in the square brackets correspond to a Feynman graph with $v-1$ vertices. This new graph is essentially the same as the graph with $v$ vertices but without a leaf and with a loop transformed into a leaf. Besides, the removal of a leaf leaves the structure of the graph invariant, meaning that it transforms a cactus in a cactus and a non-cactus in a non-cactus. 

We can iterate the computation obtaining the contribution
\begin{equation*}
\frac{N_A+N_{\bar{A}}}{2}
\end{equation*}
for each vertex. At the end of this computation the remaining term will be $\sum_k 1=N$, and this concludes the proof. 
\qed
\end{proof}
The evaluation of the contribution of non-cactus graphs is not as simple as the one of the cactus. Nevertheless, we can give an upper bound for it, by bounding the loop contributions.

\begin{theorem}
\label{thm:noncactus}
A loop gives a contribution that is lower or equal than $(N_A+N_{\bar{A}})/2$.
\end{theorem}
\begin{proof} 
We can isolate the contribution of each single loop obtaining
\begin{equation*}
[\dots, p(i)\ j',\dots, p(j)\ i',\dots]=\sum_{k_{i'},k_{j'}}\D(k_i,k_{i'};k_{p(i)},k_{j'})\D(k_{j},k_{j'};k_{p(j)},k_{i'}).
\end{equation*}
If we substitute here the expression of the coupling function in theorem~\ref{thm:delta} we find
\begin{eqnarray}
\sum_{k_{i'},k_{j'}}\delta_{k_i+k_{i'},k_{p(i)}+k_{j'}}\ \delta_{k_{j}+k_{j'},k_{p(j)}+k_{i'}}\ \delta_{0,(k_i-k_{p(i)})\wedge(k_{i'}-k_{p(i)})}\  \nonumber \\
 \times\ \delta_{0,(k_{j}-k_{p(j)})\wedge(k_{j'}-k_{p(j)})}  f(k_i-k_{p(i)},k_{i'}-k_{p(i)})f(k_{j}-k_{p(j)},k_{j'}-k_{p(j)})\nonumber\\
= \delta_{k_i+k_{j},k_{p(i)}+k_{p(j)}}\sum_{l} \delta_{0,(k_i-k_{p(i)})\wedge\ l}\ f(k_i-k_{p(i)},l)\ f(k_i-k_{p(i)},l+k_{p(i)}-k_{j})\nonumber\\
\times \sum_{k_{j'}}\delta_{k_{j}+k_{j'},k_{p(j)}+l+k_{p(i)}}\ \delta_{0,(k_{j}-k_{p(j)})\wedge(k_{j'}-k_{p(j)})}\nonumber\\
= \delta_{k_i+k_{j},k_{p(i)}+k_{p(j)}}\sum_{l}  f(k_i-k_{p(i)},l)\ f(k_i-k_{p(i)},l+k_{p(i)}-k_{j})\nonumber\\ \times \delta_{0,(k_i-k_{p(i)})\wedge\ l}\ \delta_{0,(k_i-k_{p(i)})\wedge(l+k_{p(i)}-k_{j})}.\nonumber\\
\label{eq:dim_loop}
\end{eqnarray}
It is straightforward to prove that $f(k,l)\leq 1$. Moreover, the condition that the Kronecker deltas have to be different from zero and the assumption that the binomial coefficient is zero if one of its argument is negative, remark~\ref{rmk:binomial}, fix the positions in which $l$ can be different from $0$ and give the condition:
\begin{equation*}
|l|\leq n_A=\left[\frac{n}{2}\right],
\end{equation*} 
and so the last expression in~(\ref{eq:dim_loop}) can be bounded by $N_{A}=d^{\left[\frac{n}{2}\right]}$ and thus by $(N_A + N_{\bar{A}})/2$. 
\qed
\end{proof}

\begin{figure}[t]
   \centering
   \includegraphics[scale=0.20]{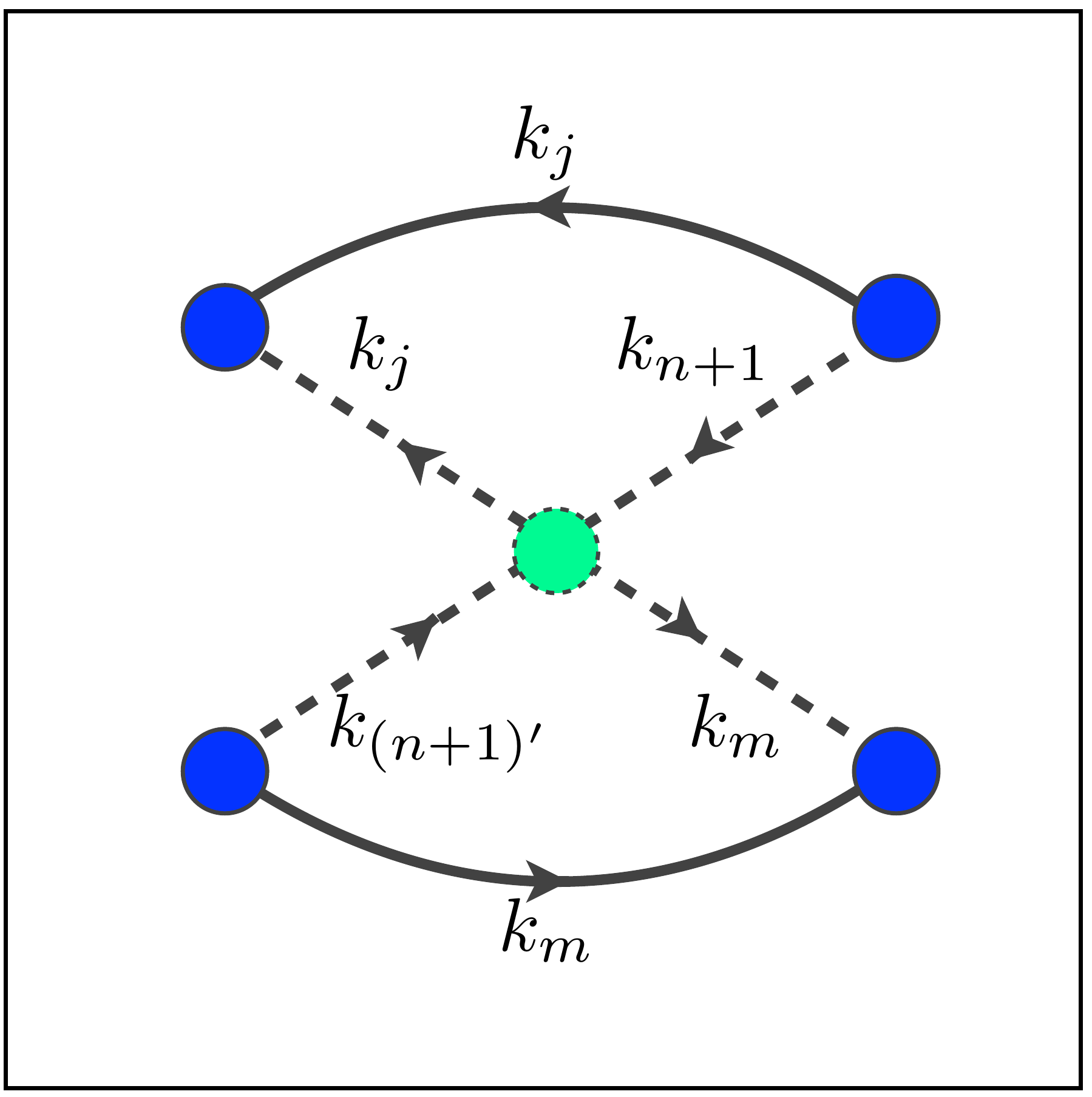}
   \caption{Pinching operation}
\label{fig:pinching}
\end{figure}

In the next part of this section we will give an explicit way to compute the degeneracy of graphs. In particular, we prove that we can compute the degeneracy of a generic $(v+1)$-vertex graph (we will call it daughter) knowing only the degeneracy $\deg$ of a $v$-vertex graph from which the graph is generated (we will call it mother). 
In the following, given a $(v+1)$-vertex graph $G$, we will call $G^{(m)}$ its mother graph. In this notation $G^{(m^2)}$ will be the mother of the mother and so on until the 1-vertex graph is obtained. 

\begin{definition}
We define \textit{pinching} the operation that connects two edges adding a vertex to the graph (see fig.~\ref{fig:pinching}).
\end{definition}
The following proposition illustrate the degeneracy of this operation.
\begin{proposition}
Adding a vertex through pinching increases the degeneracy $\deg(G)$ of a graph  by factor:
\begin{enumerate}[a)]
\item $4$ if the  four vertices are non degenerate or if they degenerate into two but the edges have different directions;
\item $2$ if the four vertices degenerate into one;
\item $2$ if the four vertices degenerate into two and the directions  of the two edges are the same.
\end{enumerate}
\end{proposition}
\begin{remark} 
From now on we suppose to start from a $v$-vertex graph and to add a vertex labeled by $(v+1,(v+1)')$.
\end{remark}
\begin{proof}
In the first case (fig.~\ref{fig:pinching}) the mother graph is of the form
\begin{equation*}
[\ \dots,j\ p(i'), \dots, m\ p(l'), \dots\ ],
\end{equation*}
then the daughter graph can be represented by:
\begin{equation*}
[\ \dots,v'\ p(i'), \dots, (v+1)\ p(l'), \dots, j\ m].
\end{equation*}
After the exchange of $v'$ and $v+1$ or $j$ and $m$ the graph is left unchanged, therefore the degeneracy of this new graph has an extra factor of $4$ compared with the degeneracy of the mother,  i. e. $\deg(G)=4\deg(G^{(m)})$.

In the second case, see fig.~\ref{fig:eye}(a), the $v$-vertex graph can be represented by 
\begin{equation*}
[\ \dots,j\ j', \dots\ ],
\end{equation*}
and the pinching leads to the representation
\begin{equation*}
[\ \dots,(v+1)\ (v+1)', \dots, j\ j' \ ],
\end{equation*}
or equivalently to
\begin{equation*}
[\ \dots,(v+1)'\ (v+1), \dots, j\ j' \ ].
\end{equation*}
Since there are no other possibilities, $\deg(G)=2\deg(G^{(m)})$.

In the last case, fig.~\ref{fig:eye}(b), we start from the graph
\begin{equation*}
[\ \dots,i\ j', \dots \ ],
\end{equation*}
and arrive to 
\begin{equation*}
[\ \dots,(v+1)\ (v+1)', \dots, i\ i' \ ],
\end{equation*}
where again we have an extra factor $2$ of degeneracy.
\qed
\end{proof}

\begin{figure}[t]
 \begin{minipage}[b]{7cm}
   \centering
   \includegraphics[scale=0.20]{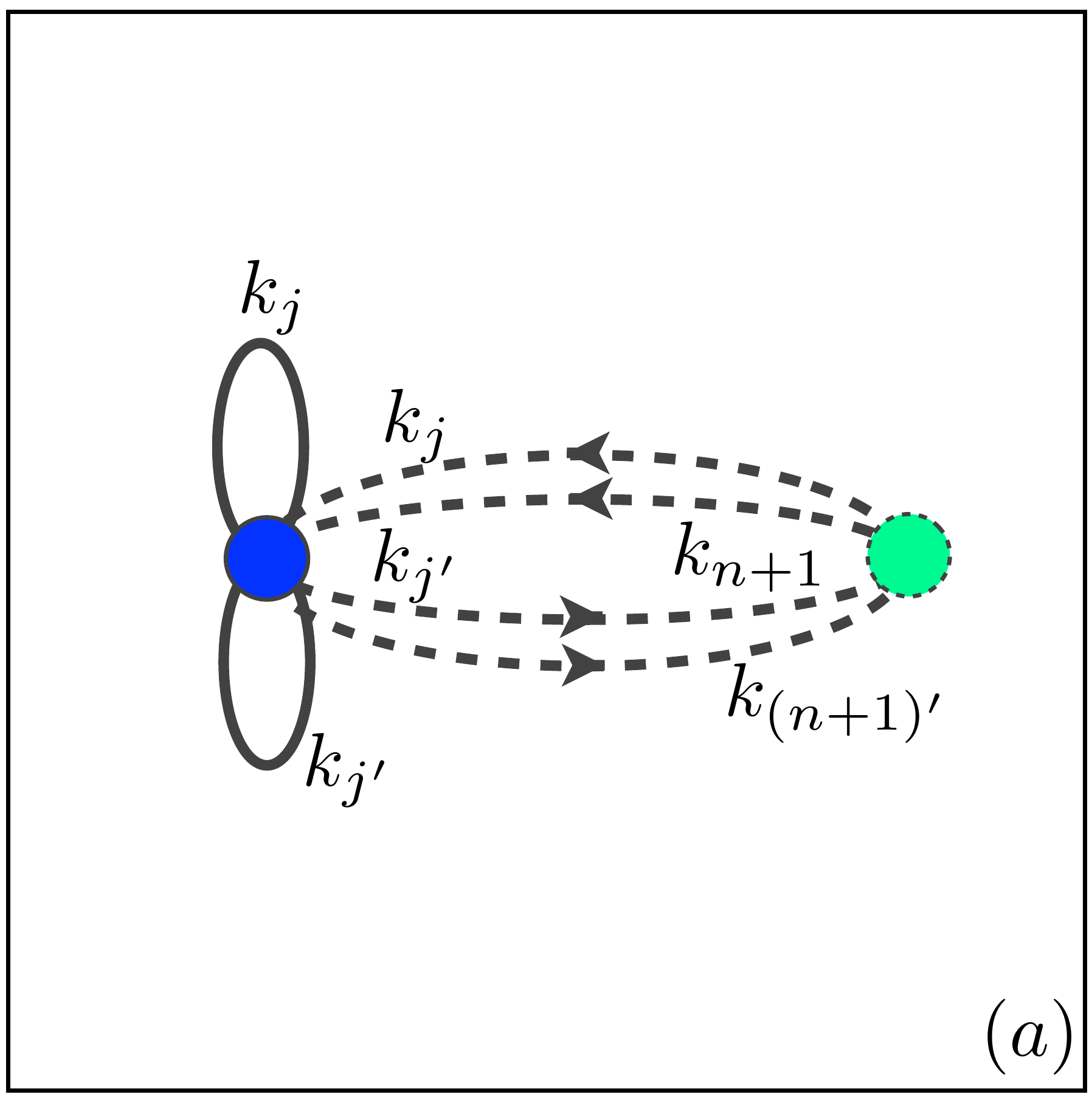}
 \end{minipage}
 \ \hspace{-15mm} 
 \begin{minipage}[b]{7cm}
  \centering
  \includegraphics[scale=0.20]{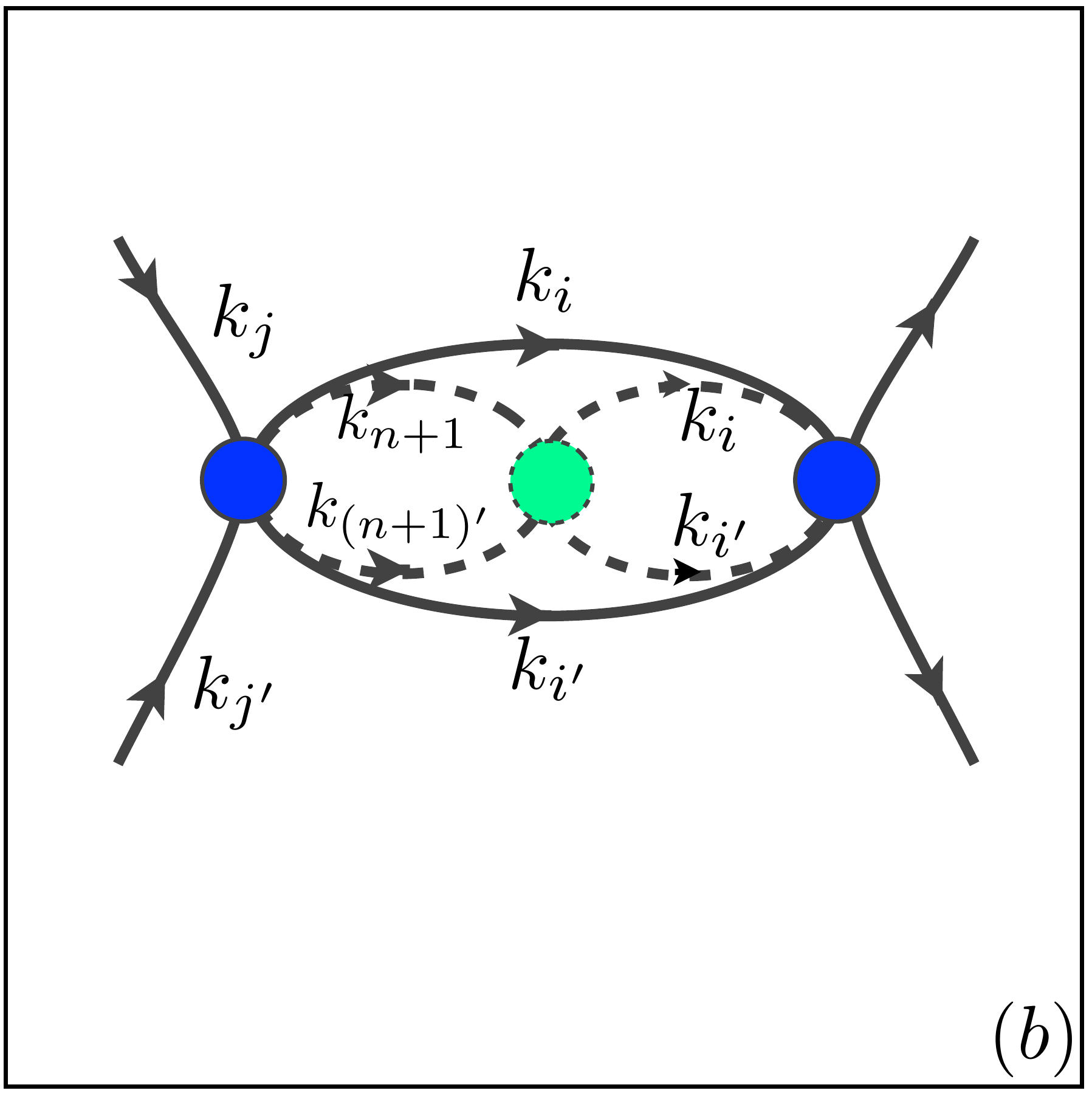}
 \end{minipage}
\caption{(a) Four vertices degenerating into one. (b) Four vertices degenerating into two.}
\label{fig:eye}
\end{figure}

\begin{remark} 
In the previous proposition there is no mention to the case in which the four vertices degenerate into two and the two edges degenerate into one. Nevertheless this operation coincide with the germination of a leaf, fig.~\ref{fig:leaf_germ}. Indeed, before the creation of the leaf the graph is associated to
\begin{equation*}
[\ \dots,i\ p(j'), \dots\ ],
\end{equation*}
while after the pinching the representation becomes 
\begin{equation*}
[\ \dots, (v+1)\ p(j'), \dots, i\ (v+1)'\ ], 
\end{equation*}
and $3$ other combinations lead to the same graph: swapping $v+1$ and $(v+1)'$ or the elements in the two pairs. 
\end{remark}

\begin{remark}
Remarkably, the pinching operation allows us to construct all $(v+1)$-vertex graphs starting from the $v$-vertex graphs and in addition provides  a practical way for computing the degeneracy. Moreover, we recall that cactus graphs are always generated by other cactus, and there is no way to transform a non-cactus into a simpler graph just by adding a vertex. However, it is not true in general that the degeneracy of the daughter graph is the degeneracy of the mother multiplied by the degeneracy of the pinching, $\deg_{\mathrm{pinc}}$. In fact, the addition of a vertex can break the symmetry of a graph, and when this happens we have a factor lower or equal than 
\begin{equation*}
v+1=\frac{(v+1)!}{v!},
\end{equation*}
so that the more the graph is symmetric the lower is its degeneracy. Furthermore, sometimes pinching different edges gives rise to the same graph so that in counting the degeneracy we have to consider also all these possibilities. 
\end{remark}

Summing up the previous considerations we can state the following. 
\begin{proposition}
The degeneracy of a graph $G$ is
\begin{center}
\begin{math}
\deg(G)=\deg(G^{(m)})\times \deg_{\mathrm{pinc}} \times \tilde{\deg}_{\mathrm{pinc}} \times (v+1),
\end{math}
\end{center}
where the last factor is the symmetrization factor and $\tilde{\deg}_{\mathrm{pinc}}=$different ways in which pinching the edges gives rise to the same graph.
\end{proposition}
\begin{remark}
Pinching edges pointing in the same direction gives rise to different graphs from the ones obtained pinching edges with the same direction. The difference between these two cases is shown in fig.~\ref{fig:3p}(b) and fig.~\ref{fig:3p}(c) where the two graphs differ in the arrows direction.
\end{remark}
\begin{figure}[t]
 \centering
   \includegraphics[scale=0.20]{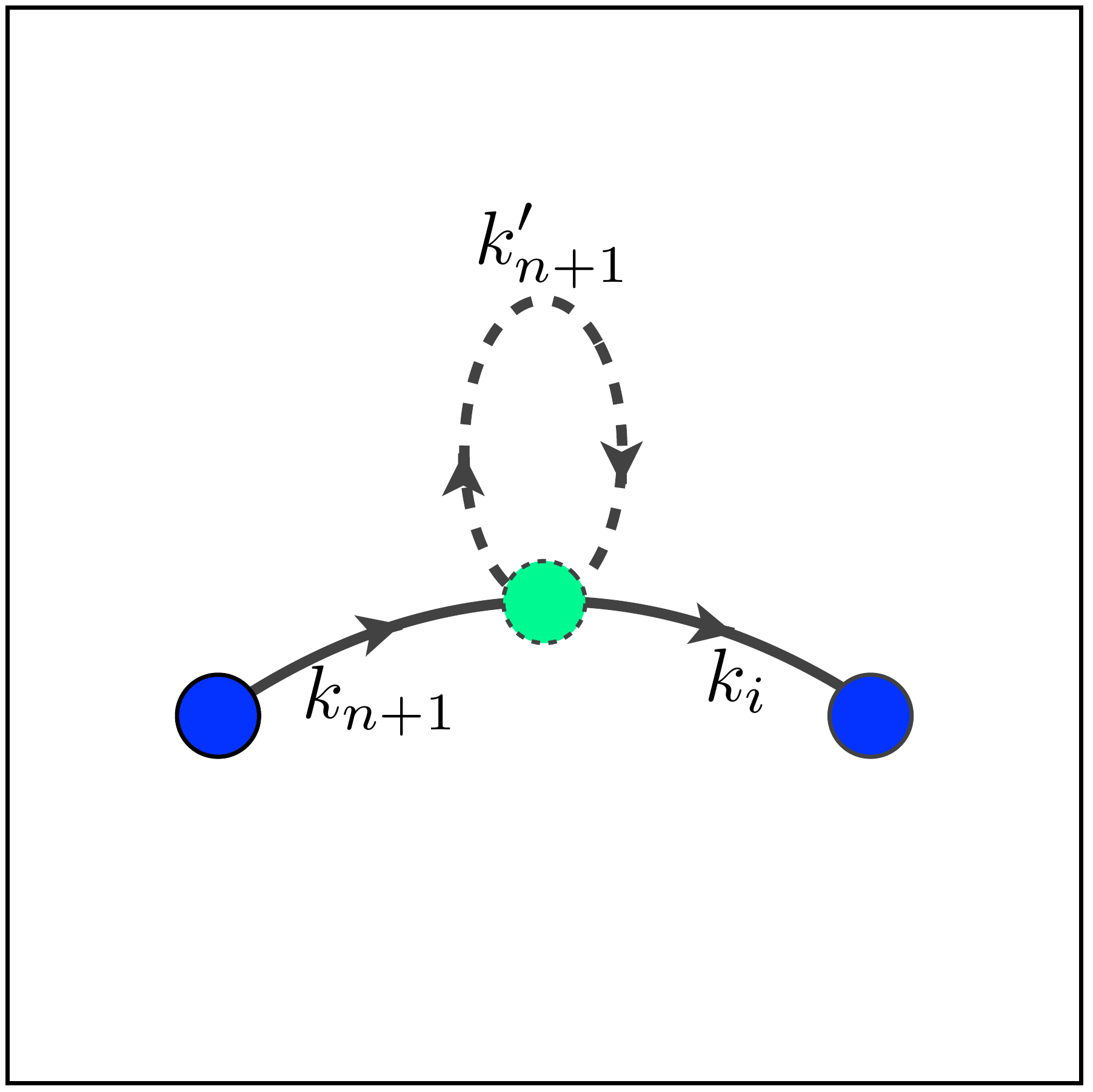}
\caption{Germination of a leaf.}
\label{fig:leaf_germ}
 \end{figure}
 We are now ready to give a bound for the degeneracy of the graphs.
\begin{theorem}
\label{thm:degeneracy}
The degeneracy of a $v$-vertex graph $G$ satisfies
\begin{equation*}
\deg(G)\leq 2^{2v}\, v!
\end{equation*}
\end{theorem}
\begin{proof}
To compute the degeneracy of a graph in the worst case scenario we have to compute all the configurations $[p(1)\ p(1'), \dots,p(v)\ p(v')]$ that lead to the same graph. In the worst case the pinching operation gives a  factor of $4$ for each point and every permutation of the vertices  leaves the configuration unchanged, so that the degeneracy, in the worst case is $4^v v!$.
\qed
\end{proof}

\begin{example}
The only graph with one vertex is the one in fig.~\ref{fig:otto}, that has degeneracy $2$. From this graph we can generate the two connected graphs with two vertices. The one in fig.~\ref{fig:cactus1} is obtained by a non degenerate pinching so that its degeneracy is $2\times 4\times 2=16$, while the one in fig.~\ref{fig:eye2} is generated by a degenerate pinching and so its degree of degeneracy is $2\times 2 =4$. 
\begin{figure}[!h]
  \centering
  \includegraphics[scale=0.20]{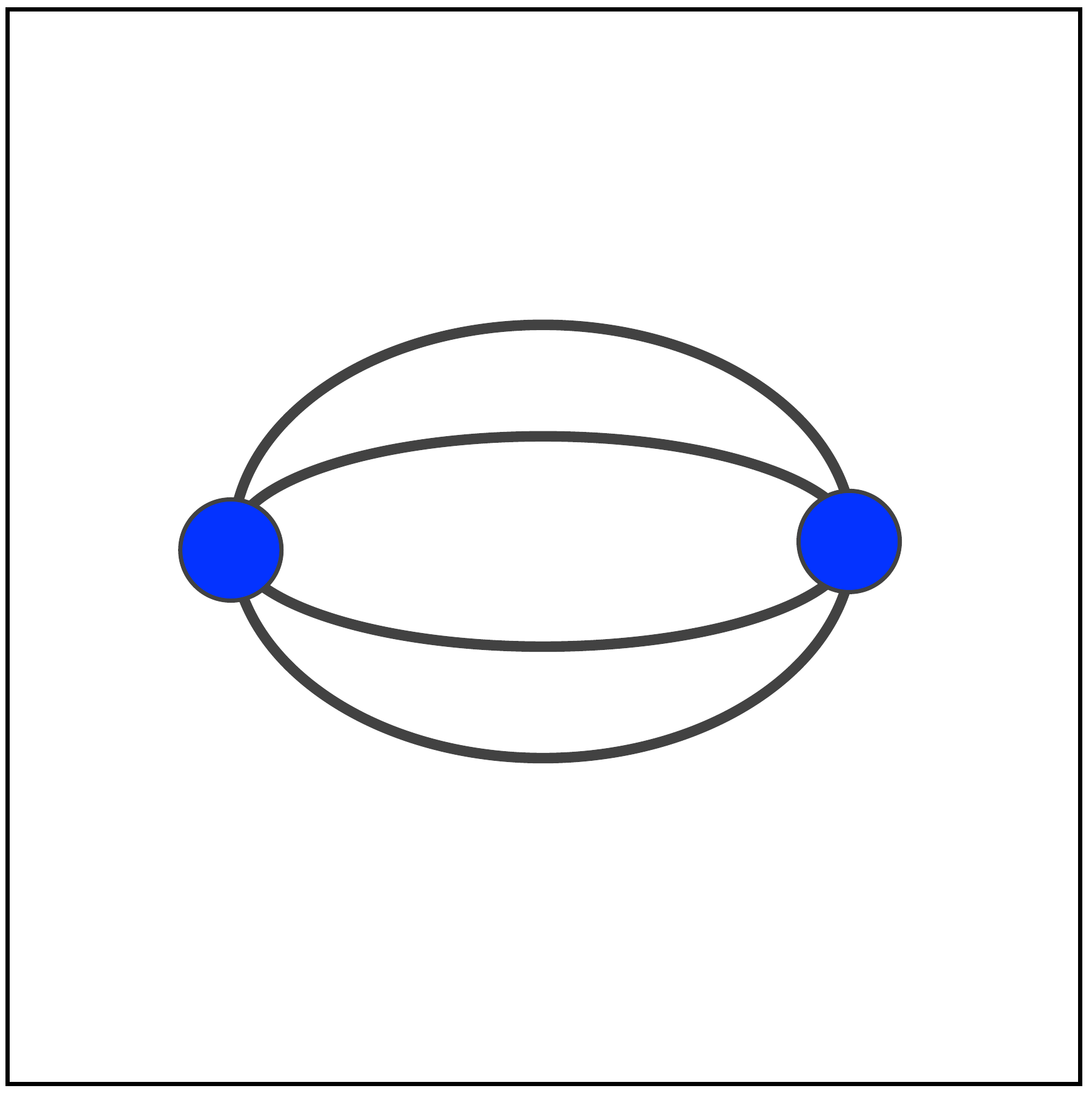}
 \caption{A non-cactus with two vertices.}
 \label{fig:eye2}
 \end{figure}
\end{example}
\begin{figure}[t]
 \begin{minipage}[b]{7cm}
   \centering
   \includegraphics[scale=0.20]{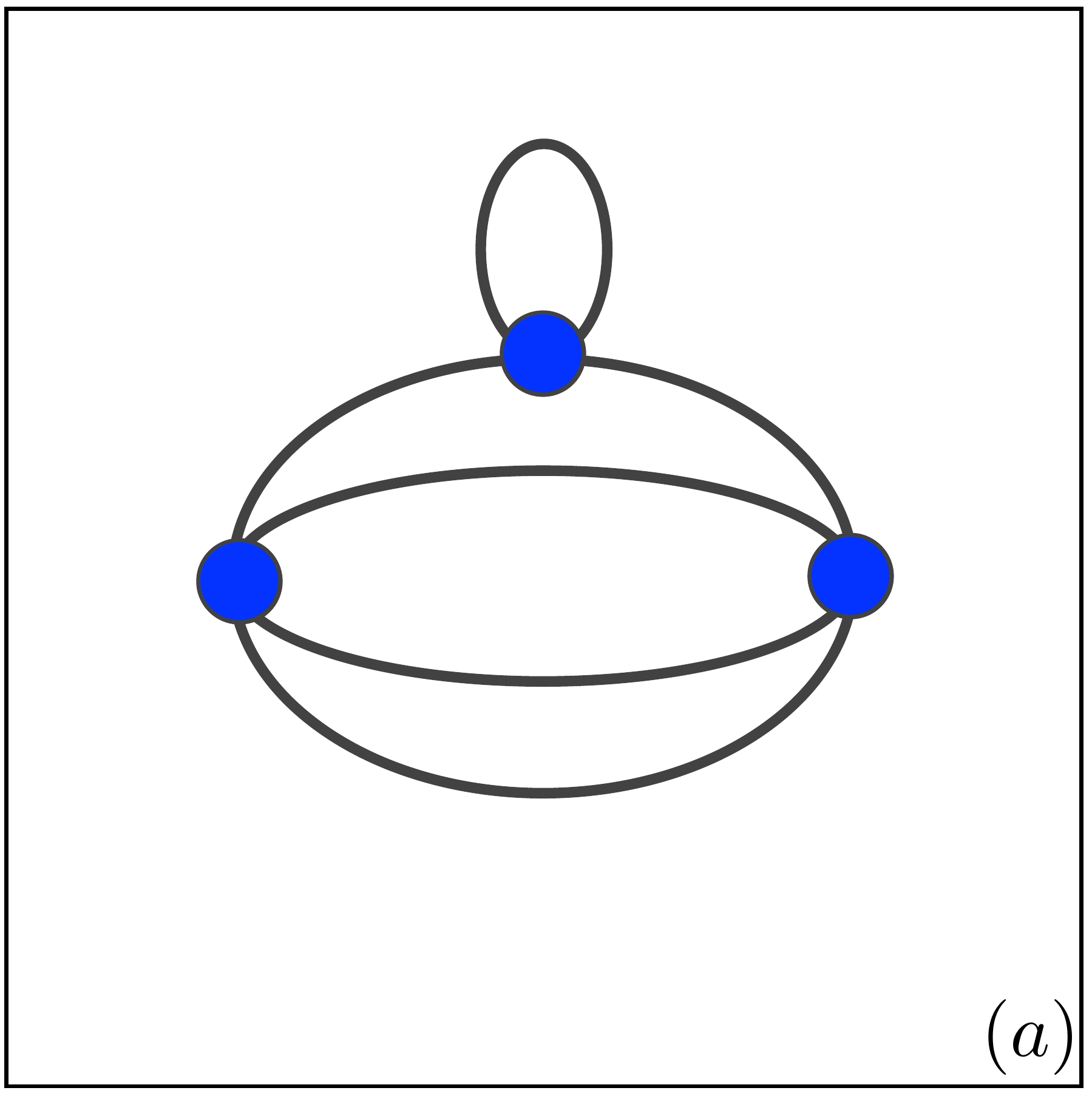}
 \end{minipage}
 \ \hspace{-15mm} 
 \begin{minipage}[b]{7cm}
  \centering
  \includegraphics[scale=0.20]{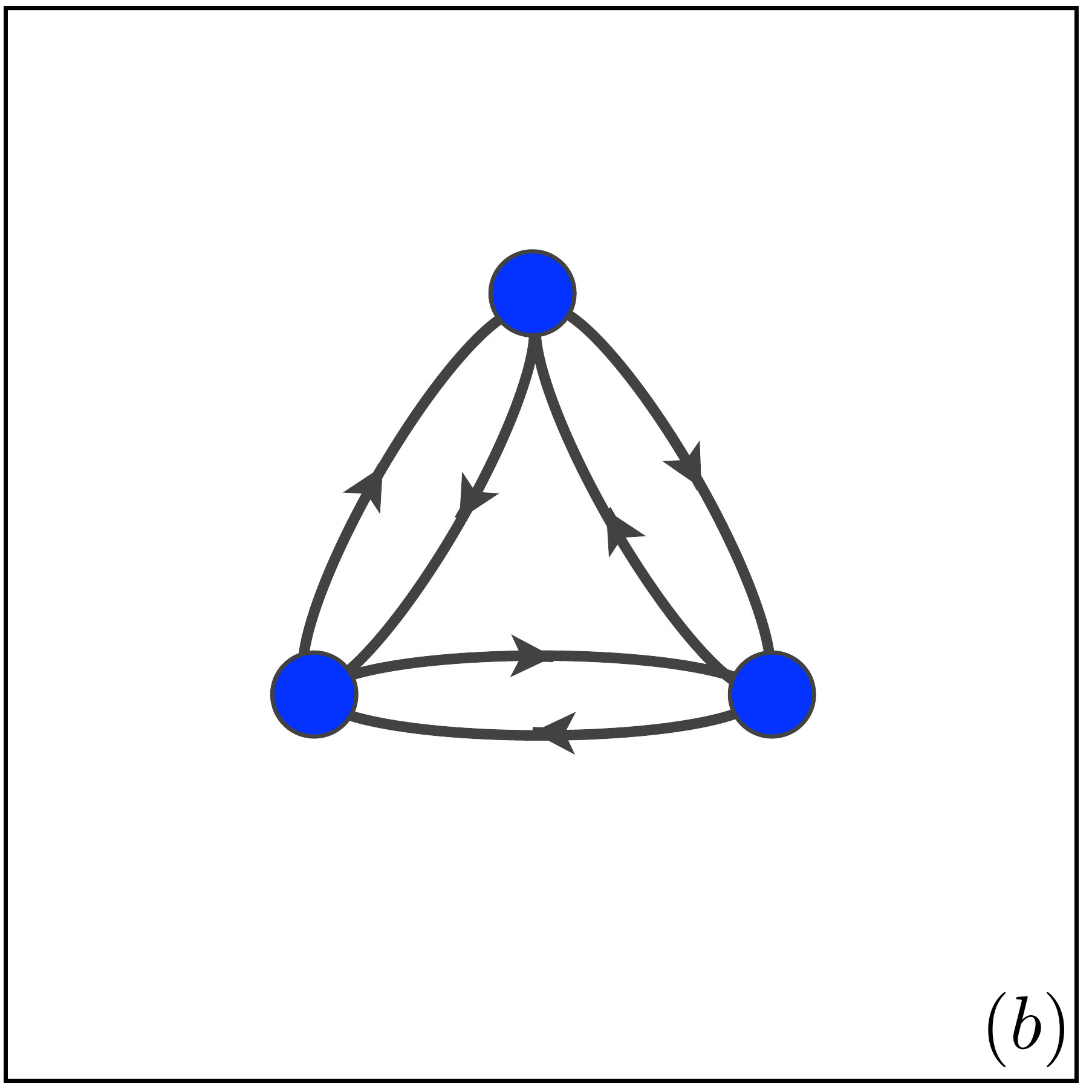}
 \end{minipage}
 \center
 \begin{minipage}[b]{7cm}
  \centering
  \includegraphics[scale=0.20]{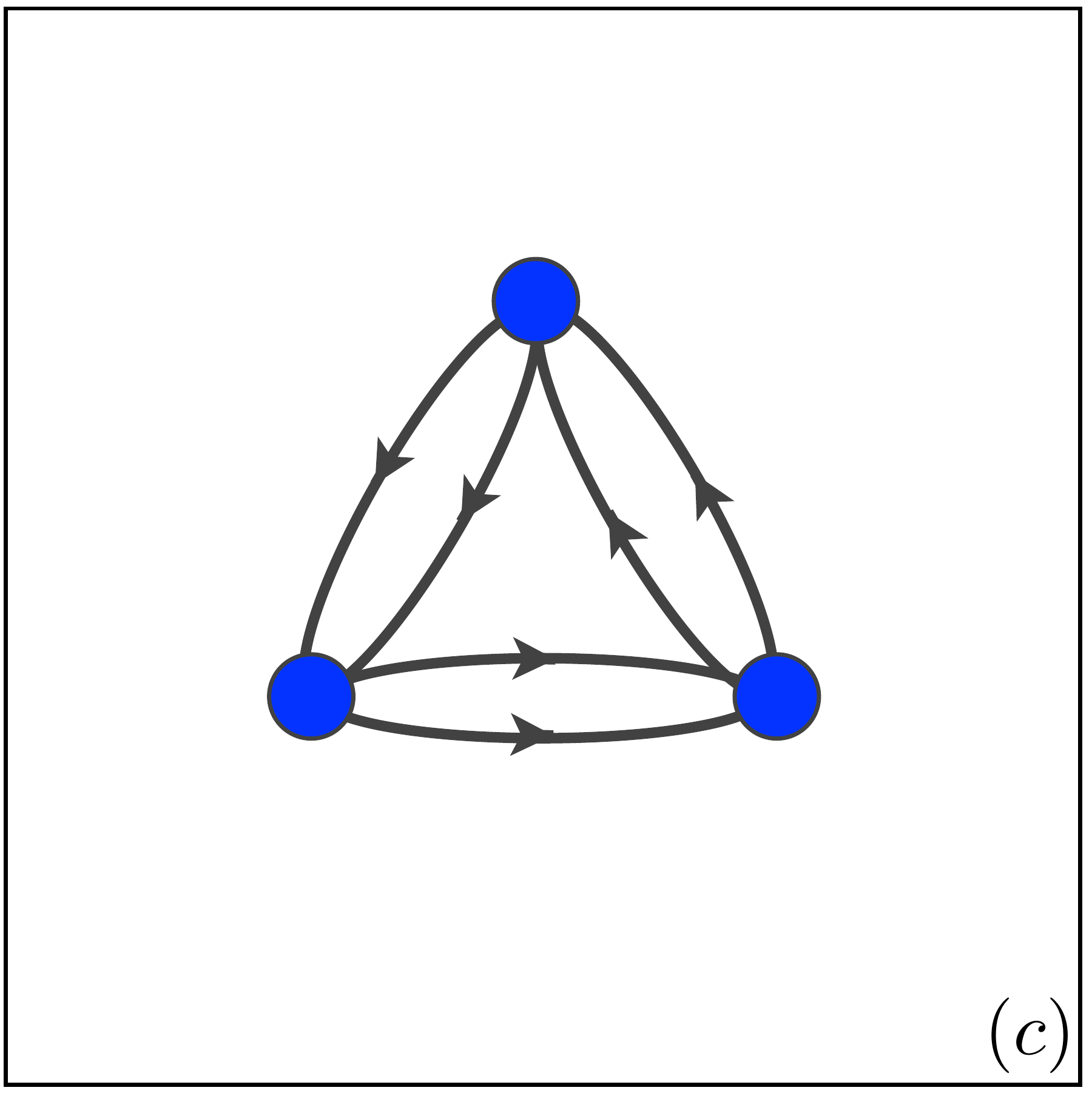}
 \end{minipage}
 \caption{All connected non-cactus with three vertices}
 \label{fig:3p}
\end{figure}
\begin{example}
The Feynman graph in fig.~\ref{fig:3p}(a) is generated from the one in fig.~\ref{fig:eye2} by the germination of a leaf.  Its degeneracy is $4\times 4\times 4\times 3=192$, where the $3$ is the symmetrization factor.
\end{example}
\begin{example}
The graph in fig.~\ref{fig:3p}(c) generates after a pinching of the graph in  fig.~\ref{fig:eye2}, so that its degeneracy is $16\times 4=64$, while the graph in fig.~\ref{fig:3p}(b) is generated  by pinching of the graph in fig.~\ref{fig:eye2}, therefore its degeneracy is $4 \times 2\times 2=16$.
\end{example}

%%%%%%%%%%%%%%%%%%%%%%%%%%%%%%%%%%%
\section{High-Temperature Expansion: proof of Theorem~\ref{thm:cactus}}
\label{sec:hightemp}
The representation in terms of Feynman graphs   introduced in the previous section gives important information about the series~\eqref{eq:series}, and in particular allows to prove theorem~\ref{thm:cactus}. 
\begin{proof}[of theorem~\ref{thm:cactus}]
From Eq.~(\ref{eq:cumulants2}) we can write $\langle H^m\rangle_0$ in terms of the square brackets:
\begin{equation*}
\langle H^m\rangle_0=\frac{1}{N(N+1)\dots(N+2m-1)} \sum_{p\in \mathrm{S}_{2m}}[p(1)p(1'),\dots, p(m)p(m')]. 
\end{equation*}
Let us now divide the permutation in $\mathrm{S}_{2m}$ in the ones that generate a cactus, $\mathcal{P}_1$, and the ones that generate a non-catus, $\mathcal{P}_2$.
\begin{eqnarray*}
\langle H^m\rangle_0 &=&\frac{1}{N(N+1)\dots(N+2m-1)}\nonumber\\
& &\times\left(\sum_{p\in \mathcal{P}_1}[p(1)p(1'),\dots, p(m)p(m')]+ \sum_{p\in \mathcal{P}_2}[p(1)p(1'),\dots, p(m)p(m')]\right) 
\nonumber\\
\end{eqnarray*}
 Applying theorem~\ref{thm:cactus}, the first permutations give:
\begin{equation*} 
[p(1)p(1'),\dots, p(m)p(m')]=\left(\frac{N_A+N_{\bar{A}}}{2}\right)^m,
\end{equation*}
recalling that the number of vertices in a Feynman graph generated by permutations in $\mathrm{S}_{2m}$ is exactly $m$, while for permutation in $\mathcal{P}_2$ applying theorem~\ref{thm:noncactus}, we have:
\begin{equation*} 
[p(1)p(1'),\dots, p(m)p(m')]\leq N \left(\frac{
N_A + N_{\bar{A}}}{2}\right)^{m-1}. 
\end{equation*}
Defining $\tilde{C}_1(m)$ and $\tilde{C}_2(m)$ as the number of nonequivalent cactus and non-cactus graphs (respectively) with $m$ vertices, and recalling that according to theorem~\ref{thm:degeneracy} every $m$-vertex graph has at most degeneracy $2^{2m} m!$, we eventually get
\begin{eqnarray*}
\langle H^m\rangle_0&\leq&\nonumber\frac{2^{2m}m!}{N(N+1)\dots(N+2m-1)}
\\
& & \times\bigg[ \tilde{C}_1(m)\ N\left(\frac{N_A+N_{\bar{A}}}{2}\right)^m+ \tilde{C}_2(m)\ N \left(\frac{
N_A+N_{\bar{A}}}{2}\right)^{m-1}\bigg],
\end{eqnarray*}
and the first part of the theorem follows.

The second statement follows immediately from the first one,  by recalling that 
$$\frac{N_A + N_{\bar{A}}}{2} =\frac{d^{\left[\frac{n}{2}\right]}+d^{\left[\frac{n+1}{2}\right]}}{2} \ge d^{\left[\frac{n}{2}\right]}. $$
This concludes the proof of theorem~\ref{thm:cactus}. 
\qed
\end{proof}

\begin{remark}
A final remark is in order. As stated at the end of Section~\ref{sec:main}, it is clear that, since $C_1(m)$ and $C_2(m)$ do not depend on $d$, in the limit $d\to \infty$ the contribution to the sum due to the presence of the non-cactus graphs goes to zero and only the presence of cactus becomes relevant. Heuristically, we can attribute the presence of frustration in the system to the relevance of the non-cactus graphs in the series.  
\end{remark}

\begin{remark}
We can apply theorem~\ref{thm:cactus} from a different perspective, keeping $d$ fixed and evaluating the limit $N=2^n \to\infty$. In particular, let us consider the case of qubits ($d=2$) where exact explicit expressions for the first, second and third moment have been obtained using Feynman diagrams \cite{cumulants} and one can evaluate the contributions from non-cactus diagrams in the limit $N=2^n \to\infty$. 

\begin{enumerate}

\item $m=1$. This case is trivial. Contributions to the moment only come from the cactus shown in fig.~\ref{fig:otto}.

\item $m=2$. An explicit calculation shows that cactus diagrams are of the form in fig.~\ref{fig:otto} (disconnected) whereas non-cactus diagrams are of the form in fig.~\ref{fig:eye2}. We obtain

 \begin{equation*}
\frac{\langle H^2\rangle_{0,\mathrm{NC}}}{\langle H^2\rangle_{0,\mathrm{C}}} = \frac{f_2(N)}{(N+4)(N_A+N_{\bar{A}})^2},
\end{equation*}
with
\begin{eqnarray*}
f_2(N)= 2
\left(\!\!\begin{array}{c} n    \\ n_{A} \end{array}\!\!\right)^{\!\!\!-1}
\!\!\!\sum_{0\leq k\leq n_A}\!\!
\left(\!\!\begin{array}{c} n_A    \\ k \end{array}\!\!\right)
\left(\!\!\begin{array}{c} n_{\bar{A}}    \\ k \end{array}\!\!\right)
2^{n/2}\left[4^{n/4-k}+4^{-(n/4-k)}\right] .
\label{eq:f2appprim}
\end{eqnarray*}
It is possible to prove \cite{cumulants} that, in the limit
$N\to\infty$,
\begin{equation*}
\label{eq:f2}
f_2(N)\sim 3 \sqrt{2} N^\alpha,
\end{equation*}
with
\begin{equation*}
\alpha=\log_2 3-1 \simeq 0.5850.
\label{eq:alphadef}
\end{equation*}
Therefore, for $N\to\infty$ and $N_A\simeq N_{\bar{A}}=2^{\left[\frac{n}{2}\right]}$ we
have
\begin{eqnarray*} \label{eq:secondocum}
& & \frac{\langle H^2\rangle_{0,\mathrm{NC}}}{\langle H^2\rangle_{0,\mathrm{C}}}\sim \frac{b(2)}{N^{2-\alpha}}\le \frac{C(2)}{2^{\left[\frac{n}{2}\right]}},
\end{eqnarray*}
with $b(2)$ a positive constant.

\item $m=3$. We obtain
\begin{eqnarray*}
\frac{\langle H^3\rangle_{0,\mathrm{NC}}}{\langle H^3\rangle_{0,\mathrm{C}}} = \frac{16f_3^{(1)}(N)+64f_3^{(0)}(N)+ 3 (N+8)(N_A+N_{\bar{A}})f_2(N)}{(40+12N+N^2)(N_A+N_{\bar{A}})^3}
\end{eqnarray*}
where the role of $f_3^{(1)}$ and $f_3^{(0)}$ is analogous to that of $f_2$ for $m=2$. Their complete expressions are not transparent and can be found in \cite{cumulants}. We notice that the contributions of $f_3^{(0)},f_3^{(1)}$ come from non-cactus diagrams in fig.~\ref{fig:3p} $(b)$ and $(c)$, respectively. The contribution of $f_2$ comes from the non-cactus (connected) diagram in fig.~\ref{fig:3p} $(a)$ and from (disconnected) contributions obtained from the non-cactus in fig.~\ref{fig:eye2} and the diagram in fig.~\ref{fig:otto}. 
In the limit $N\to\infty$, we obtain
\begin{equation*}
\label{eq:f3}
f_3^{(0)}(N)\sim c N^{5-\gamma},\quad f_3^{(1)}(N)\sim N^{\alpha}
\end{equation*}
with
\begin{equation*}
\gamma \simeq 4.1583, \quad c\simeq 1.05385.
\end{equation*}
Finally, for $N\to\infty$ and $N_A\simeq N_{\bar{A}}=2^{\left[\frac{n}{2}\right]}$ we have
\begin{eqnarray*} \label{eq:terzocum}
& & \frac{\langle H^3\rangle_{0,\mathrm{NC}}}{\langle H^3\rangle_{0,\mathrm{C}}}\sim \frac{b(3)}{N^{3-\alpha}}\le \frac{C(3)}{2^{\left[\frac{n}{2}\right]}},
\end{eqnarray*}
with $b(3)$ a positive constant.

\end{enumerate}

\end{remark}

%%%%%%%%%%%%%%%%%%%%%%%%%%%%%%%%%%%
\appendix

\section{Lower Bound on the $7$-qubit potential of multipartite entanglement}
\label{app:sstate}

The minimum of the potential of multipartite entanglement for $7$ qubits is not known yet. Up to now only guesses have been proposed and some numerical bounds have been found. Here we construct a $7$-qubit state with the lowest $\pi_{\mathrm{ME}}$ found until now, to the best of our knowledge. The state can be constructed starting from an orthonormal basis of 3- and 4-qubit MMES. 

Considering the computational basis, we will write the $i$-th vector of the basis in terms of the Fourier coefficients $z^{(i)}=(z^{(i)}_k)$:
\begin{equation*}
\ket{\psi}_i=\sum_{k\in\mathbb{Z}_2^n} z^{(i)}_k\ket{k}.
\end{equation*}

The 3-qubit MMES basis ($d=2,n=3$) is made of GHZ states $\ket{\mathrm{GHZ}}_i$ ($i=0,..,7$) with Fourier coefficients
\begin{eqnarray*}
z^{(0)}&=&\frac{1}{\sqrt{2}}\left(1,0,0,0,0,0,0,1\right); \qquad z^{(1)}=\frac{1}{\sqrt{2}}\left(1,0,0,0,0,0,0,-1\right);\\
z^{(2)}&=&\frac{1}{\sqrt{2}}\left(0,1,0,0,0,0,1,0\right);\qquad z^{(3)}=\frac{1}{\sqrt{2}}\left(0,1,0,0,0,0,-1,0\right);\\
z^{(4)}&=&\frac{1}{\sqrt{2}}\left(0,0,1,0,0,1,0,0\right);\qquad z^{(5)}=\frac{1}{\sqrt{2}}\left(0,0,1,0,0,-1,0,0\right);\\
z^{(6)}&=&\frac{1}{\sqrt{2}}\left(0,0,0,1,1,0,0,0\right);\qquad z^{(7)}=\frac{1}{\sqrt{2}}\left(0,0,0,1,-1,0,0,0\right).
\end{eqnarray*} 

The 4-qubit MMES basis ($d=2,n=4$) is made by the states $\ket{\mathrm{MMES}_4}_i$ ($i=0,..,15$) with Fourier coefficients
\begin{eqnarray*}
z^{(0)}&=&\frac{1}{4}\left(-1,-1,-1,-1,-1,-1,1,1,-1,1,-1,1,1,-1,-1,1\right);\\
z^{(1)}&=&\frac{1}{4}\left(-1,-1,-1,-1,-1,-1,1,1,1,-1,1,-1,-1,1,1,-1\right);\\
z^{(2)}&=&\frac{1}{4}\left(-1,-1,-1,-1,1,1,-1,-1,-1,1,-1,1,-1,1,1,-1\right);\\
z^{(3)}&=&\frac{1}{4}\left(-1,-1,-1,-1,1,1,-1,-1,1,-1,1,-1,1,-1,-1,1\right);\\
z^{(4)}&=&\frac{1}{4}\left(-1,-1,1,1,-1,-1,-1,-1,-1,1,1,-1,1,-1,1,-1\right);\\
z^{(5)}&=&\frac{1}{4}\left(-1,-1,1,1,-1,-1,-1,-1,1,-1,-1,1,-1,1,-1,1\right);\\
z^{(6)}&=&\frac{1}{4}\left(-1,-1,1,1,1,1,1,1,-1,1,1,-1,-1,1,-1,1\right);\\
z^{(7)}&=&\frac{1}{4}\left(-1,-1,1,1,1,1,1,1,1,-1,-1,1,1,-1,1,-1\right);
\end{eqnarray*} 
\begin{eqnarray*} 
z^{(8)}&=&\frac{1}{4}\left(-1,1,-1,1,-1,1,1,-1,-1,-1,-1,-1,1,1,-1,-1\right);\\
z^{(9)}&=&\frac{1}{4}\left(-1,1,-1,1,-1,1,1,-1,1,1,1,1,-1,-1,1,1\right);\\
z^{(10)}&=&\frac{1}{4}\left(-1,1,-1,1,1,-1,-1,1,-1,-1,-1,-1,-1,-1,1,1\right);\\
z^{(11)}&=&\frac{1}{4}\left(-1,1,-1,1,1,-1,-1,1,1,1,1,1,1,1,-1,-1\right);\\
z^{(12)}&=&\frac{1}{4}\left(-1,1,1,-1,-1,1,-1,1,-1,-1,1,1,1,1,1,1\right);\\
z^{(13)}&=&\frac{1}{4}\left(-1,1,1,-1,-1,1,-1,1,1,1,-1,-1,-1,-1,-1,-1\right);\\
z^{(14)}&=&\frac{1}{4}\left(-1,1,1,-1,1,-1,1,-1,-1,-1,1,1,-1,-1,-1,-1\right);\\
z^{(15)}&=&\frac{1}{4}\left(-1,1,1,-1,1,-1,1,-1,1,1,-1,-1,1,1,1,1\right).
\end{eqnarray*} 

We will look for a minimizing 7-qubit state expressed in terms of tensor products of the elements of the two basis
\begin{equation*}
\ket{\sigma_7}=\sum_{i,j} c_{i,j} \ket{\mathrm{MMES}_4}_i\otimes \ket{\mathrm{GHZ}}_j.
\end{equation*}
We have numerically evaluated the complex coefficients $c_{i,j}$ so that the state $\ket{\sigma_7}$ is a minimizer of the potential of multipartite entanglement $\pi_{\mathrm{ME}}$.  
We have found a solution such that the only non-vanishing $c_{i,j}$'s are tabulated in Table~\ref{tab:coefficients} and expressed in the form
\begin{equation*}
c_{i,j}=\phi_{i,j}^{(re)}+i\phi_{i,j}^{(im)},
\end{equation*}
with $\phi_{i,j}^{(re)},\phi_{i,j}^{(im)}$ denoting the real and imaginary part of the coefficient, respectively.
\begin{table}[h!]
\centering
\begin{tabular}{||c | c | c ||} 
 \hline
$(i,j)$ & $\phi_{i,j}^{(re)}$ & $\phi_{i,j}^{(re)}$  \\ [0.5ex] 
 \hline\hline
 (1,1)&0.313685 & -0.019416 \\
(1,4) &-0.124963 & 0.00751404 \\
 (2,2)&\text{6.16876805$\times 10^{-6}$} & -0.000116371 \\
 (2,3)&-0.000103808 & -0.0000691072 \\
 (3,2)&0.000046695 & -0.0000735369 \\
 (3,3)&-0.000243151 & -0.000195018 \\
 (4,1)&0.0193888 & 0.313752 \\
 (4,4)&-0.00777771 & -0.124766 \\
 (5,5)&0.0719262 & 0.15313 \\
 (5,8)&0.152837 & -0.0721254 \\
 (6,6)&-0.160604 & -0.0526771 \\
 (6,7)&0.0528744 & -0.160803 \\
 (7,6)&0.0527307 & -0.160861 \\
 (7,7)&0.161076 & 0.0527179 \\
 (8,5)&0.153033 & -0.0719374 \\
 (8,8)&-0.0723194 & -0.153041 \\
 (9,5)&0.0529309 & -0.160616 \\
 (9,8)&0.160688 & 0.0526321 \\
 (10,6)&-0.072288 & -0.153269 \\
 (10,7)&-0.15302 & 0.0720842 \\
 (11,6)&-0.152985 & 0.0719888 \\
 (11,7)&0.0719157 & 0.153028 \\
 (12,5)&-0.161016 & -0.0527083 \\
 (12,8)&0.0526094 & -0.160931 \\
 (13,1)&0.0128427 & -0.0812437 \\
 (13,4)&0.032297 & -0.204478 \\
 (14,2)&-0.087611 & -0.17006 \\
 (14,3)&0.264942 & 0.00134756 \\
 (15,2)&-0.245851 & -0.124495 \\
 (15,3)&-0.132874 & 0.115682 \\
 (16,1)&-0.0128813 & 0.0812742 \\
 (16,4)&-0.0323628 & 0.204452 \\ [1ex] 
 \hline
\end{tabular}
\caption{Table of non-vanishing coefficients $c_{i,j}$ obtained from the minimization of $\pi_{\mathrm{ME}}$ in the case of 7 qubits.}
\label{tab:coefficients}
\end{table}
The value of $\pi_{\mathrm{ME}}$ in this case is
\begin{equation*}
\pi_{\mathrm{ME}}(\sigma_7)=0.131952.
\end{equation*} 
In figure~\ref{fig:istos} we plot the histogram characterizing the distribution of the purities among the balanced bipartitions. 
\begin{figure}[!h]
\centering
\includegraphics[scale=0.70]{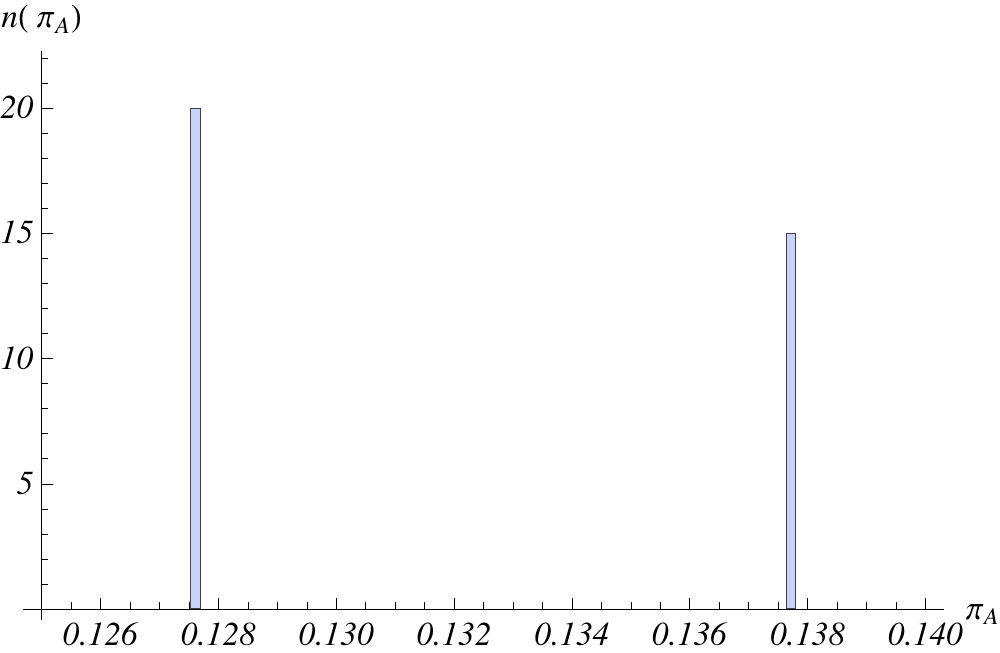}
\caption{Purity distribution among balanced bipartitions for the $7$-qubit state $\ket{\sigma_7}$.}
\label{fig:istos}
\end{figure} 
An interesting feature of this state is that all bipartitions are frustrated, i.e. none of them reaches the minimum for the corresponding purity but the distribution is fairly well peaked in correspondence of two spikes only.

%%%%%%%%%%%%%%%%%%%%%%%%%%%%%%%%%%%

\section{Maximum distance separable codes and perfect MMES}
\label{app:codth}

The definition of a (classical) code starts from the choice of a set of elements $\Sigma$ that constitutes the alphabet of the code. Even if there are no restriction on this choice, since information theory is constructed around machines and computers, the most common set considered is $\Sigma=\{0,1\}$, with the clear meaning that the information is encoded in bits. 
\begin{definition}
Given an alphabet $\Sigma$, a \textit{(classical) code} $C$ is a set of strings, called \textit{codewords}, over $\Sigma$ of fixed length. 
\end{definition}

In the set of all possible codewords of fixed length $n$, $\Sigma^n$, it is possible to define a distance in the following.  
\begin{definition}
The \textit{Hamming distance} between two strings of the same length, $\mathrm{d}_H:\Sigma^n\times\Sigma^n\to\R$, is the number of positions in which the corresponding symbols are different.
\end{definition}
\begin{example}
The Hamming distance between the two strings $0011$ and $1001$ is $\d_H(0011,1001)=3$.
\end{example}
\begin{remark}
The Hamming distance is a proper distance, i.e. it is positive, symmetric and satisfies the triangle inequality. 
\end{remark}
\begin{definition}
The \textit{minimal Hamming distance} of the code, $\delta$, is defined as
\begin{equation*}
\delta=\min_{\{v,w\in C,v\neq w\}}\d_H(v,w). 
\end{equation*}
\end{definition}

The following theorem gives a bound on the dimension of the code (the maximum number of codewords) and the minimal Hamming distance of the codewords. 
\begin{theorem}[Singleton bound~\cite{singleton}]
\label{lm:code}
For any code $C\subseteq \Sigma^n$ the following inequality holds:
\begin{equation*}
M \geq q^{n-\delta+1},
\end{equation*}
with $M$ the number of the codewords and $\delta$ its minimum Hamming distance.
\end{theorem}
\begin{definition}
A code for which the the Singleton bound is saturated is called \textit{Maximum Distance Separable} (MDS) code.
\end{definition}

Let us consider a code $C=\{c_j\}$, with $N_A$ codewords of length $n$ and alphabet $\mathbb{Z}_d$. Using the codewords of $C$ we can construct the $n$-qudit state:
\begin{equation}
\label{eq:MMES}
\ket{\psi}=\frac{1}{\sqrt{N_A}}\sum_{j=1}^{N_A}\ket{c_j}.
\end{equation}
If the minimal Hamming distance of $C$ is greater than $n_A+1$, after the partial trace over a balanced bipartition all the off-diagonal terms, $\tr_{\bar{A}}(\ketbra{c_j}{c_k})$, vanish. Indeed
\begin{equation*}
\tr_{\bar{A}}(\ketbra{c_j}{c_k})=\sum_{l\in\mathbb{Z}_d^{n_{\bar{A}}}}\braket{l}{c_j}\braket{c_k}{l}\neq 0,
\end{equation*}
if and only if $c_j$ and $c_k$ have at least $n_{\bar{A}}$ symbols in common. Moreover, the presence of $d^{\left[\frac{n}{2}\right]}$ terms in the sum is due to the necessity of having $\rho_A$ proportional to identity, i. e.  $\frac{\I}{{N_A}}$ for every bipartition $(A,\bar{A})$. Therefore, for this state, $\pi_{ME}$ reaches its minimum and the state in eq.~\eqref{eq:MMES} is a perfect MMES. 

It remains to prove the existence of such a code. In particular, from the Singleton bound, $\delta\geq n_{\bar{A}}+1$, meaning that we are addressing the relation between $n$ and $d$ in order for a MDS code to exist.   

\begin{theorem}
If $d$ is a prime or a prime power, a MDS code exists if $n\leq d-1$.
\end{theorem}
The MDS codes to which the theorem is referring are the Reed-Solomon codes~\cite{reedsolomon}. This means that given $n$ it is always possible to choose the first suitable $d\geq n+1$ in order to construct a perfect MMES. 

\begin{remark}
This bound gives only a bound on the point at which frustration disappears. Indeed, in the case of $5$ qubits, for example, a perfect MMES exists, while according to the bound we need $d\geq 4$.
\end{remark}

\begin{acknowledgements}
We thank Giorgio Parisi and Saverio Pascazio for useful discussions and insightful comments.
P. Facchi and S. Di Martino are supported by the Gruppo Nazionale per la Fisica Matematica (GNFM) of the Istituto Nazionale di Alta Matematica (INdAM). G. Florio is supported by MIUR through the project ``VirtualMurgia'' and GNFM through ``Progetto Giovani''. P. Facchi and G. Florio are supported by Istituto Nazionale di Fisica Nucleare through the project ``QUANTUM''. S. Di Martino is supported by the ERC (Advanced Grant IRQUAT, project number ERC-267386).
\end{acknowledgements}

\end{document}